\documentclass[mlabstract,onecolumn]{jmlr}

\usepackage{longtable}

\usepackage{booktabs}
\usepackage[load-configurations=version-1]{siunitx} 

\theorembodyfont{\upshape}
\theoremheaderfont{\scshape}
\theorempostheader{:}
\theoremsep{\newline}

\DeclareMathOperator*{\argmin}{arg\,min}

\DeclareMathOperator*{\bij}{Bij}
\DeclareMathOperator*{\supp}{supp}
\DeclareMathOperator*{\im}{Im}
\DeclareMathOperator*{\ib}{\text{IB}}

\newcommand{\tq}{\tilde{q}}

\newcommand{\Xcal}{\mathcal{X}}
\newcommand{\Ycal}{\mathcal{Y}}
\newcommand{\Scal}{\mathcal{S}}
\newcommand{\Tcal}{\mathcal{T}}

\newcommand{\hl}[1]{#1}
\newcommand{\hll}[1]{#1}
\newcommand{\hlll}[1]{#1}
\newcommand{\hlc}[1]{#1}
\newcommand{\hlcc}[1]{#1}
\newcommand{\tb}[1]{#1}
\newcommand{\del}{}

\jmlrvolume{}
\firstpageno{1}
\editors{Sophia Sanborn, Christian Shewmake, Simone Azeglio, Nina Miolane}

\jmlryear{2023}
\jmlrworkshop{Symmetry and Geometry in Neural Representations}

\title[Towards IT-Based Discovery of Equivariances]{Towards Information Theory-Based\titlebreak Discovery of Equivariances}

  \author{\Name{Hippolyte Charvin} \Email{h.charvin@herts.ac.uk}\\
  \Name{Nicola {Catenacci Volpi}} \Email{n.catenacci-volpi@herts.ac.uk}\\
  \Name{Daniel Polani} \Email{d.polani@herts.ac.uk}\\
  \addr Adaptive Systems Research Group, University of Hertfordshire}

\begin{document}

\maketitle

\begin{abstract}
    The presence of symmetries imposes a stringent set of constraints on a system. This constrained structure allows intelligent agents interacting with such a system to drastically improve the efficiency of learning and generalization, through the internalisation of the system's symmetries into their information-processing. In parallel, principled models of complexity-constrained learning and behaviour make increasing use of information-theoretic methods. Here, we wish to marry these two perspectives and understand whether and in which form the information-theoretic lens can ``see'' the effect of symmetries of a system. For this purpose, we propose a novel variant of the Information Bottleneck principle  which has  served as a productive basis for many principled studies of learning and information-constrained adaptive behaviour. We show (in the discrete case \hlc{and under a specific technical assumption}) that our approach formalises a certain duality between symmetry and information parsimony: namely, channel equivariances can be characterised by the optimal \emph{mutual information-preserving joint compression} of the channel's input and output. This information-theoretic treatment furthermore suggests a principled notion of ``soft'' equivariance, whose ``coarseness'' is measured by the amount of input-output mutual information preserved by the corresponding optimal compression. This new notion offers a bridge between the field of bounded rationality and the study of symmetries in neural representations. The framework may also allow  (exact and soft) equivariances to be automatically discovered.
\end{abstract}

\bigskip

\bigskip

\section{Introduction}

Our work is motivated by a programme of formalising the relationship between the presence of coherent structures in an environment, and the informational efficiency that these structures make possible for an (artificial or biological) agent that learns and interacts with them. Our intuition is that there is a fundamental duality between structure and information: in short, any structure in a system affords a possibility of informational efficiency to an agent interacting with it, and every improvement in an agent's informational efficiency must exploit some kind of structure in the system it interacts with. 

As a first step towards the operationalisation of this intuition, we focus on a specific kind of structure: symmetries, and, more precisely, the \emph{equivariances} of probabilistic channels \citep{bloem-reddyProbabilisticSymmetriesInvariant2020}.
We seek to first design a formal method to identify the duality between equivariances and information, and will leave the modeling of concrete systems to future work. Previous results \citep{achilleEmergenceInvarianceDisentanglement2018} exhibited links between invariance extraction and the Information Bottleneck (IB) method \citep{tishbyInformationBottleneckMethod2000}, which optimally compresses one variable under the constraint of preserving information about a second variable. \tb{Here, we} adapt this idea to the more general context of equivariances, which increasingly appear crucial to efficient learning and generalisation \citep{higginsSymmetryBasedRepresentationsArtificial2022}. We propose \hlll{an extension} of the IB method whose solutions indeed characterise \hlc{--- under a specific technical assumption ---} the equivariances of discrete probabilistic channels. This characterisation provides, as far as we are aware, a novel and intuitively appealing point of view on equivariances, through the notion of \emph{mutual information-preserving \hll{optimal} joint compression} \hl{of the channel's input and output. Namely, our result characterises equivariances as the pairs of transformations \hll{made indiscernible from the identity} by such a compression.}

However,  to eventually grasp real-world symmetries, which might be much less stringent than mathematical equivariances in the classic, ``exact'' sense, we need to consider  ``soft'' notions of equivariance. The problem then arises of \emph{how to measure the ``divergence''} from being an exact equivariance. Here, we build on our new characterisation of exact equivariances to define the ``coarseness'' of soft equivariances through the resolution of the informationally optimal compression that they make possible. \hl{Namely, soft equivariances of ``granularity'' $\lambda$ are defined as pairs of transformations \hll{made indiscernible from the identity} by an optimal compression which \emph{partially} preserves the channel's input-output mutual information, to a degree specified by $\lambda$.}

This information-theoretic point of view on equivariances links the study of symmetries in biological and artificial agents to the field of bounded rationality \citep{geneweinBoundedRationalityAbstraction2015}, through the duality between informationally optimal representations and the corresponding extracted equivariances.  But crucially, this method might also allow one to \emph{discover} soft equivariances: we will sketch a roadmap towards computing equivariances as defined here.

\paragraph{Assumptions and notations:} We fix \hll{finite sets} $\mathcal{X}$ and $\mathcal{Y}$ and a \emph{fully supported} probability $p(X,Y)$ on $\mathcal{X} \times \mathcal{Y}$.\footnote{For now, we work under the hypothesis that in real-world scenarios, there will typically be at least some noise spillover into all possible configurations. We leave to future work a generalisation to non-fully supported $p(X,Y)$ \hl{(see Remark \ref{rmk:proof_for_general_pYgX} in Appendix \ref{apd:end_proof})} \hll{and to non-finite $p(X,Y)$ (see Appendix \ref{apd:generalisation_to_nonfinite_case})}.} ``Bottlenecks'' are variables $T$ defined on $\mathcal{T} := \mathbb{N}$. The probability simplex defined by a \hll{finite set} $\mathcal{A}$ is denoted by $\Delta_{\mathcal{A}}$. Conditional probabilities, also called channels, will often be regarded as \hl{functions between probability simplices, or as linear maps between vector spaces (e.g., a channel from $\{1, \dots, n\}$ to itself can be regarded as a function from $\Delta_{\{1, \dots, n\}}$ to itself, or as linear map from $\mathbb{R}^n$ to itself)}. The set of channels with input space $\mathcal{A}$ and output space $\mathcal{B}$, resp. output space $\mathcal{A}$ itself, are denoted by $C(\mathcal{A}, \mathcal{B})$, resp. $C(\mathcal{A})$. The set of bijections of $\mathcal{A}$ is $\bij(\mathcal{A})$, and for  $\gamma \in \bij(\mathcal{A})$, $a \in \mathcal{A}$, we write $\gamma \cdot a := \gamma(a)$. The identity map on $\mathcal{A}$ is written $e_{\mathcal{A}}$. The symbol $\circ$ denotes function composition, resp. channel composition, depending on the context (functions are seen as deterministic channels when they are composed with another channel). The symbol $\delta_P$ means $1$ when the proposition $P$ is true, and $0$ otherwise. $D(\cdot||\cdot)$ is the Kullback-Leibler divergence.

\section{The Intertwining Information Bottleneck and exact equivariances}
\label{section:iib_exact_equivariances}


\begin{definition}  \label{def:exact_channel_equivariance}
    An (exact) \emph{equivariance} of the channel $p(Y|X)$ is a pair of deterministic permutations $(\sigma, \tau) \in \bij(\mathcal{X}) \times \bij(\mathcal{Y})$ such that $p(Y|X) \circ \sigma = \tau \circ p(Y|X)$. An \emph{invariance} of $p(Y|X)$ is some $\sigma \in \bij(\mathcal{X})$ such that $p(Y|X) \circ \sigma = p(Y|X)$.
\end{definition}
It can be easily verified that the set of equivariances of $p(Y|X)$ is a group for the relation $(\sigma, \tau) \cdot (\sigma', \tau') := (\sigma \circ \sigma', \tau \circ \tau')$. This group will be called the \emph{equivariance group} of $p(Y|X)$, and be denoted $G_{p(Y|X)}$. Now, in the IB method, which, as mentioned above, has been suggested to extract channel invariances, one considers a pair of variables $X$ and $Y$, but the compressed variable is a function of only one of them, say $X$, whereas it preserves information about the second variable $Y$. This is consistent with the idea that the IB might extract invariances, because the latter transform only the space $\mathcal{X}$. However, equivariances clearly transform \emph{both} spaces $\mathcal{X}$ and $\mathcal{Y}$, so that a compression that has any hope of extracting these equivariances should be a function of \emph{both} $X$ and $Y$. For the same reason, it does not seem natural that, here, the preserved information should be either only that about $X$, or only that about $Y$. Rather, we want to formalise the following intuition: the presence of (exact, resp. soft) equivariances of $p(X,Y)$ should correspond to the possibility of compressing the joint variable $(X,Y)$ in a way that (fully, resp. partially) preserves the \emph{mutual} information $I(X;Y) := D(p(X,Y)||p(X)p(Y))$. Thus we propose to consider what we call the \emph{Intertwining Information Bottleneck} (IIB), defined for every $0 \leq \lambda \leq I(X;Y)$:
\begin{align}   \label{eq:def_intertwining_ib}
  \argmin_{\substack{\kappa \, \in \, C(\mathcal{X} \times \mathcal{Y}, \mathcal{T}) \, : \\ D(\kappa(p(X,Y)) || \kappa(p(X)p(Y))) \, = \, \lambda}} 
  I_\kappa(X,Y;T),
\end{align}
where the mutual information $I_\kappa(X,Y;T)$ is computed from the distribution $p(x,y) \kappa(t|x,y)$. 
The constraint in \eqref{eq:def_intertwining_ib} means that the channel $\kappa$ must conserve the divergence between $p(X,Y)$ and its split version $p(X)p(Y)$, to the level specified by $\lambda$. On the other hand, the minimisation of $I_\kappa(X,Y;T)$ means that $\kappa$ implements, under the latter constraint, an optimal compression. In particular, the solutions to \eqref{eq:def_intertwining_ib} for $\lambda = I(X;Y)$ formalise the intuition of largest possible compression of the pair $(X,Y)$ that still preserves the mutual information between these variables. \tb{Importantly}, \hll{both the IB and the Symmetric IB \citep{slonimMultivariateInformationBottleneck2006} can be recovered from the IIB problem} by adding the right constraint on the shape of $\kappa$ in \eqref{eq:def_intertwining_ib}. \hll{If we add the requirement that $\kappa$ can only compress the $\mathcal{X}$ coordinate, we recover the IB problem with source $X$ and relevancy $Y$; while if we rather impose that $\kappa$ must compress $\mathcal{X}$ and $\mathcal{Y}$ separately, we recover the Symmetric IB problem (see Appendix \ref{apd:iib_generalises_ib_and_sib}).}

Given the structural similarity between \eqref{eq:def_intertwining_ib} and the IB problem, the algorithms for computing the latter might be adaptable to the former. In particular, we leave to future work to \hl{prove the convergence of, and implement, an adapted version of the Blahut-Arimoto algorithm used for the IB \citep{tishbyInformationBottleneckMethod2000}. Another possibility would be to identify, and optimise for, variational bounds \citep{alemiDeepVariationalInformation2019} on the information quantities from \eqref{eq:def_intertwining_ib}. Note that} for  $\lambda = I(X;Y)$, the set of solutions can be computed explicitly, and, up to trivial transformations, it consists of a unique deterministic clustering (see Corollary \ref{cor:solutions_to_IIB_fullsupport} in Appendix \ref{apd:explicit_solution_iib_Ixy}).  Let us now formalise our intuition of duality between the (exact) equivariance group $G_{p(Y|X)}$ and the information compression that the latter makes possible.


\begin{theorem}   \label{th:char_equivariance_with_iib_Ixy}
  \hlc{Assume that $p(X)$ is such that $p(Y) := \sum_x p(Y|x) p(x)$ is uniform}, and let $\kappa \in C(\mathcal{X} \times \mathcal{Y}, \mathcal{T})$ be a solution to the IIB problem for $\lambda = I(X;Y)$. Then a pair $(\sigma,\tau) \in \bij(\mathcal{X}) \times \bij(\mathcal{Y})$ is an equivariance of $p(Y|X)$ if and only if
  \begin{align}   \label{eq:charac_equivariance_matrix_form}
     \kappa \circ (\sigma \otimes \tau) = \kappa.
  \end{align}
\end{theorem}

\begin{proof}
  See Appendix \ref{apd:proof_charac_equivariance_with_iib_solutions}.
\end{proof}

Intuitively, the essentially unique solution $\kappa$ to the IIB for $\lambda=I(X;Y)$ is the deterministic coarse-graining of the product space $\mathcal{X} \times \mathcal{Y}$ satisfying the following property: a pair of permutations $(\sigma,\tau) \in \bij(\mathcal{X}) \times \bij(\mathcal{Y})$ is an equivariance of $p(X,Y)$ if and only if this coarse-graining ``filters out'' the effect of simultaneously transforming $\mathcal{X}$ with $\sigma$ and $\mathcal{Y}$ with $\tau$ \hll{--- thus making the pair $(\sigma,\tau)$ indiscernible from the identity on $\mathcal{X} \times \mathcal{Y}$}. In particular \hlcc{--- under the theorem's assumption of uniform $p(Y)$ ---} the equivariance group of $p(X,Y)$ is characterised by the optimal compression of the joint variable $(X,Y)$ that still preserves the mutual information $I(X;Y)$. \del

The \hlcc{assumption that there exists} an input distribution $p(X)$ such that $\sum_x p(Y|x) p(x)$ is uniform means, geometrically, that the set of output distributions $\{ p(Y|x), \ x \in \mathcal{X} \}$ contains the uniform distribution in its convex hull. Clearly, this assumption is not satisfied for a generic channel $p(Y|X)$. It is however satsified, e.g., if for every output symbol $y \in \mathcal{Y}$, there is an input symbol $x \in \mathcal{X}$ such that the pointwise conditional probability $p(Y|x)$ is close to the Dirac distribution $\delta_y$. This latter condition means, intuitively, that every output symbol is achieved with high probability with a well-chosen input symbol: i.e., that the channel's noise is small. We leave to future work the question of whether the conclusion of Theorem \ref{th:char_equivariance_with_iib_Ixy} can be obtained with more general assumptions.



\section{Towards soft equivariances discovery}
\label{section:IIB_and_soft_equivariances}

To soften the notion of channel equivariance, we first allow the transformations on resp. $\mathcal{X}$ and $\mathcal{Y}$ to be \hlll{non-invertible and} stochastic. But more importantly, we have to choose \emph{the right notion of ``divergence''} from the \hlcc{exact equivariance in Definition \ref{def:exact_channel_equivariance}} being achieved. Following the dual point of view developed in Section \ref{section:iib_exact_equivariances}, we assume, intuitively, that soft equivariances should be characterised by an optimal compression of $(X,Y)$ under the constraint of, here, \emph{partially} preserving $I(X;Y)$. To make the statement precise, let us define, for $\mu \in C(\mathcal{X})$ and $\eta \in C(\mathcal{Y})$, the tensor product $ \mu \otimes \eta (x',y'|x,y) := \mu(x'|x) \eta(y'|y)$. 
\begin{definition}  \label{def:soft_equivariance}
  Let $p(Y|X)$ be given, \hlc{such that there exists some $p(X)$ yielding a uniform $p(Y) := \sum_x p(Y|x) p(x)$. For $p(X,Y)$ defined through the latter $p(X)$ and $p(Y|X)$}, let $\kappa$ be a solution to the \hlc{corresponding} IIB problem \eqref{eq:def_intertwining_ib} with parameter $0 \leq \lambda \leq I(X;Y)$. A \emph{$(\lambda, \kappa)$-equivariance} of $p(X,Y)$ is a pair $(\mu, \eta) \in C(\mathcal{X}) \times C(\mathcal{Y})$ such that
  \begin{align}   \label{eq:def_soft_equivariance}
    \kappa \circ (\mu \otimes \eta) = \kappa .
  \end{align}
  We will also call a pair $(\mu, \eta)$ a \emph{$\lambda$-equivariance} if there exists some solution $\kappa$ to the IIB problem \eqref{eq:def_intertwining_ib}, with parameter $\lambda$, such that $(\mu, \eta)$ is a $(\lambda, \kappa)$-equivariance.
\end{definition}

Intuitively, \tb{a pair $(\mu, \eta)$ is a $(\lambda, \kappa)$-equivariance if the channel $\kappa$, which implements a joint optimal compression of $X$ and $Y$ under the constraint of partially preserving their mutual information, ``filters out''} the simultaneous stochastic transformations of $\mathcal{X}$ through $\mu$ and $\mathcal{Y}$ through $\eta$ \hlll{--- thus making $(\mu,\eta)$ indiscernible from the identity on $\mathcal{X} \times \mathcal{Y}$}. Moreover, it is clear from Theorem \ref{th:char_equivariance_with_iib_Ixy} that \hlc{(under the assumption of this theorem)} exact equivariances are $\lambda$-equivariances with $\lambda=I(X;Y)$.

For fixed $\lambda$ and corresponding $\kappa$, the set of $(\lambda,\kappa)$-equivariances is clearly a \emph{semigroup} \hlcc{with respect to channel composition}. Intuitively, we expect this semigroup to get larger when $\lambda$ decreases: indeed, the IIB channel $\kappa$ then enforces a larger compression of $X$ and $Y$, thus allowing more  transformations $\mu \otimes \eta$ of $\mathcal{X} \times \mathcal{Y}$ to be \hl{``filtered out''} by this compression. More precisely, equation \eqref{eq:def_soft_equivariance} is equivalent to $\im(\mu \otimes \eta - e_{\mathcal{X} \times \mathcal{Y}}) \subseteq \ker(\kappa)$,\footnote{\hl{Here, the discrete conditional probabilities are seen as transition matrices acting on real vectors.}} and we conjecture that the dimension of $\ker(\kappa)$ increases for decreasing $\lambda$, thus allowing it to contain the image of more tranformations of the form $\mu \otimes \eta - e_{\mathcal{X} \times \mathcal{Y}}$. Note for instance that for $\lambda=0$, the IIB solutions are the channels $\kappa$ such that $\kappa(T|x,y)$ does not depend on $(x,y)$. Their kernel is the direction of the whole simplex $\Delta_{\mathcal{X} \times \mathcal{Y}}$, so that the corresponding set of $(0, \kappa)$-equivariances is the whole of $C(\mathcal{X}) \otimes C(\mathcal{Y})$.


Now, assuming that a solution $\kappa$ to the IIB is known, how can we explicitly compute the corresponding $(\lambda, \kappa)$-equivariances? The equation \eqref{eq:def_soft_equivariance} which defines soft equivariances is a polynomial equation, made of quadratic homogeneous polynomials --- more precisely, linear combinations of elements of the form $\mu_{x',x} \eta_{y',y}$. To this homogeneous polynomial equation, we must add the requirement that $\mu$ and $\eta$ are conditional probabilities: i.e., they must satisfy the linear equations $\sum_{x'} \mu_{x',x} = 1$ and $\sum_{y'} \eta_{y',y} = 1$ for all $x \in \mathcal{X}$, $y \in \mathcal{Y}$, along with the linear inequalities defining the non-negativity constraints. Overall, the pair of real matrices $(\mu,\eta)$ that satisfy the conditions of Definition \ref{def:soft_equivariance} thus correspond to the intersection of the positive orthant $\{ \forall x,x' \in \mathcal{X}, \, \forall y, y' \in \mathcal{Y}, \ \mu_{x',x} \geq 0, \: \eta_{y',y} \geq 0\}$ with the solutions of a degree 2 polynomial system of equations. We leave to future work a more involved study of this problem, and of algorithms that might solve it. 

\hll{As a first step for assessing the relevance of our method to equivariance discovery, \hll{one could also} study scenarios where specific exact equivariances are known, and verify that IIB solutions do \hlll{``filter them out'' --- in the sense of equation \eqref{eq:def_soft_equivariance}. If this is the case, one could then perturb the channel $p(Y|X)$, and investigate whether the exact equivariances of the unperturbed channel are still soft equivariances of the perturbed channel --- still in the sense of equation \eqref{eq:def_soft_equivariance}.}}



\smallskip

In short, in this work we have formalised the duality between channel equivariances and the informational efficiency that they make possible for capturing the relationship between the channel's input and output \hlc{--- under a specific technical assumption, see Theorem \ref{th:char_equivariance_with_iib_Ixy}}. We achieved this with a novel extension of the IB principle, which leads to a principled generalisation of \tb{exact} equivariances into ``soft'' ones. The proposed approach might help understand the emergence of symmetries in neural systems through the lens of information parsimony, and potentially opens a new path towards the automatic discovery of exact and soft equivariances. 

\paragraph{Funding}

H.C. and D.P. were funded by the Pazy Foundation under grant ID 195.

\paragraph{Acknowledgements}

\hlll{D.P. thanks} Naftali Tishby for early discussions leading to the present studies.

\bibliography{bibliography.bib}

\begin{thebibliography}{15}
\providecommand{\natexlab}[1]{#1}
\providecommand{\url}[1]{\texttt{#1}}
\expandafter\ifx\csname urlstyle\endcsname\relax
  \providecommand{\doi}[1]{doi: #1}\else
  \providecommand{\doi}{doi: \begingroup \urlstyle{rm}\Url}\fi

\bibitem[Achille and Soatto(2018)]{achilleEmergenceInvarianceDisentanglement2018}
Alessandro Achille and Stefano Soatto.
\newblock Emergence of {{Invariance}} and {{Disentanglement}} in {{Deep Representations}}.
\newblock pages 1--9, February 2018.
\newblock \doi{10.1109/ITA.2018.8503149}.

\bibitem[Alemi et~al.(2019)Alemi, Fischer, Dillon, and Murphy]{alemiDeepVariationalInformation2019}
Alexander~A. Alemi, Ian Fischer, Joshua~V. Dillon, and Kevin Murphy.
\newblock Deep {{Variational Information Bottleneck}}, October 2019.
\newblock Comment: 19 pages, 8 figures, Accepted to ICLR17.

\bibitem[Ay et~al.(2017)Ay, Jost, L{\^e}, and Schwachh{\"o}fer]{ayInformationGeometry2017}
Nihat Ay, J{\"u}rgen Jost, H{\^o}ng~V{\^a}n L{\^e}, and Lorenz Schwachh{\"o}fer.
\newblock \emph{Information {{Geometry}}}, volume~64 of \emph{Ergebnisse Der {{Mathematik}} Und Ihrer {{Grenzgebiete}} 34}.
\newblock {Springer International Publishing}, {Cham}, 2017.
\newblock ISBN 978-3-319-56477-7 978-3-319-56478-4.
\newblock \doi{10.1007/978-3-319-56478-4}.

\bibitem[Billingsley(1995)]{billingsleyProbabilityMeasure1995}
Patrick Billingsley.
\newblock \emph{Probability and Measure}.
\newblock Wiley Series in Probability and Mathematical Statistics. {Wiley}, {New York, NY}, 3. ed edition, 1995.
\newblock ISBN 978-0-471-00710-4.

\bibitem[{Bloem-Reddy} and Teh(2020)]{bloem-reddyProbabilisticSymmetriesInvariant2020}
Benjamin {Bloem-Reddy} and Yee~Whye Teh.
\newblock Probabilistic symmetries and invariant neural networks.
\newblock \emph{J. Mach. Learn. Res}, 21:\penalty0 61, January 2020.
\newblock ISSN 1532-4435.

\bibitem[Csisz{\'a}r and K{\"o}rner(2011)]{csiszarInformationTheoryCoding2011}
Imre Csisz{\'a}r and J{\'a}nos K{\"o}rner.
\newblock \emph{Information {{Theory}}: {{Coding Theorems}} for {{Discrete Memoryless Systems}}}.
\newblock {Cambridge University Press}, {Cambridge}, 2 edition, 2011.
\newblock ISBN 978-0-521-19681-9.
\newblock \doi{10.1017/CBO9780511921889}.

\bibitem[Genewein et~al.(2015)Genewein, Leibfried, {Grau-Moya}, and Braun]{geneweinBoundedRationalityAbstraction2015}
Tim Genewein, Felix Leibfried, Jordi {Grau-Moya}, and Daniel Braun.
\newblock Bounded {{Rationality}}, {{Abstraction}}, and {{Hierarchical Decision-Making}}: {{An Information-Theoretic Optimality Principle}}.
\newblock \emph{Frontiers in Robotics and AI}, 2, November 2015.
\newblock \doi{10.3389/frobt.2015.00027}.

\bibitem[{Gilad-Bachrach} et~al.(2003){Gilad-Bachrach}, Navot, and Tishby]{gilad-bachrachInformationTheoreticTradeoff2003}
Ran {Gilad-Bachrach}, Amir Navot, and Naftali Tishby.
\newblock An {{Information Theoretic Tradeoff}} between {{Complexity}} and {{Accuracy}}.
\newblock In Gerhard Goos, Juris Hartmanis, Jan Van~Leeuwen, Bernhard Sch{\"o}lkopf, and Manfred~K. Warmuth, editors, \emph{Learning {{Theory}} and {{Kernel Machines}}}, volume 2777, pages 595--609. {Springer Berlin Heidelberg}, {Berlin, Heidelberg}, 2003.
\newblock ISBN 978-3-540-40720-1 978-3-540-45167-9.
\newblock \doi{10.1007/978-3-540-45167-9_43}.

\bibitem[Gray(2014)]{grayEntropyInformationTheory2014}
Robert~M. Gray.
\newblock \emph{Entropy and {{Information Theory}}}.
\newblock {Springer New York, NY}, 2 edition, September 2014.
\newblock ISBN 978-1-4899-8132-5.

\bibitem[Higgins et~al.(2022)Higgins, Racani{\`e}re, and Rezende]{higginsSymmetryBasedRepresentationsArtificial2022}
Irina Higgins, S{\'e}bastien Racani{\`e}re, and Danilo Rezende.
\newblock Symmetry-{{Based Representations}} for {{Artificial}} and {{Biological General Intelligence}}.
\newblock \emph{Frontiers in Computational Neuroscience}, 16, 2022.
\newblock ISSN 1662-5188.

\bibitem[Kallenberg(2017)]{kallenbergRandomMeasuresTheory2017}
Olav Kallenberg.
\newblock \emph{Random {{Measures}}, {{Theory}} and {{Applications}}}, volume~77 of \emph{Probability {{Theory}} and {{Stochastic Modelling}}}.
\newblock {Springer International Publishing}, {Cham}, 2017.
\newblock ISBN 978-3-319-41596-3 978-3-319-41598-7.
\newblock \doi{10.1007/978-3-319-41598-7}.

\bibitem[Lemar{\'e}chal(2001)]{lemarechalLagrangianRelaxation2001}
Claude Lemar{\'e}chal.
\newblock Lagrangian {{Relaxation}}.
\newblock In Michael J{\"u}nger and Denis Naddef, editors, \emph{Computational {{Combinatorial Optimization}}: {{Optimal}} or {{Provably Near-Optimal Solutions}}}, pages 112--156. {Springer Berlin Heidelberg}, {Berlin, Heidelberg}, 2001.
\newblock ISBN 978-3-540-45586-8.
\newblock \doi{10.1007/3-540-45586-8_4}.

\bibitem[Rudin(1987)]{rudinRealComplexAnalysis1987}
Walter Rudin.
\newblock \emph{Real and {{Complex Analysis}}}.
\newblock {McGraw-Hill, Inc.}, January 1987.

\bibitem[Slonim et~al.(2006)Slonim, Friedman, and Tishby]{slonimMultivariateInformationBottleneck2006}
Noam Slonim, Nir Friedman, and Naftali Tishby.
\newblock Multivariate {{Information Bottleneck}}.
\newblock \emph{Neural Computation}, 18\penalty0 (8):\penalty0 1739--1789, August 2006.
\newblock ISSN 0899-7667, 1530-888X.
\newblock \doi{10.1162/neco.2006.18.8.1739}.

\bibitem[Tishby et~al.(2000)Tishby, Pereira, and Bialek]{tishbyInformationBottleneckMethod2000}
Naftali Tishby, Fernando~C. Pereira, and William Bialek.
\newblock The information bottleneck method, April 2000.

\end{thebibliography}

\appendix

\section{Relation between IB, Symmetric IB and Intertwining IB}
\label{apd:iib_generalises_ib_and_sib}

\hlcc{In this appendix as in other ones, we will omit the subscript ``$q$'' in $I_q(X,Y;T)$, or in similar informational quantities that depend on $q = q(T|X,Y)$, when it does cause any confusion.} \hl{Let us start with the following lemma, which will prove useful below:
\begin{lemma}   \label{lemma:equality_can_be_inequality_iib}
  Let $f$ and $g$ be continuous real functions defined on a convex subspace $C$ of a \hll{topological} vector space, such that $g$ is convex \hll{and non-negative, the image of $g$ contains $0$,} and $g^{-1}(0) \subseteq f^{-1}(0)$. Let $\lambda \geq 0$, and consider the constrained optimisation problem
  \begin{align}   \label{eq:def_general_convex_concave_optimisation_problem}
    \argmin_{\substack{v \, \in \, C \, : \\ f(v) \, \geq \, \lambda}} \, g(v).
  \end{align}
  Then every solution $v$ to \eqref{eq:def_general_convex_concave_optimisation_problem} \hll{(i.e., every minimiser of \eqref{eq:def_general_convex_concave_optimisation_problem}) }must satisfy $f(v)= \lambda$. \hlll{In other words}, the set of solutions to \eqref{eq:def_general_convex_concave_optimisation_problem} coincides with the set of solutions to \hll{
  \begin{align*}
    \argmin_{\substack{v \, \in \, C \, : \\ f(v) \, = \, \lambda}} \, g(v).
  \end{align*}}
\end{lemma}
}

\begin{proof}
\hll{If $f$ is bounded from above by $\lambda$,} \hlll{then a solution $v$ to \eqref{eq:def_general_convex_concave_optimisation_problem} must satisfy both $f(v) \geq \lambda$ and $f(v) \leq \lambda$, so that $f(v) = \lambda$ and the proof is done.} Let us thus consider a vector $v \in C$ such that $f(v) > \lambda$, \hlll{and fix also some $v_0 \in g^{-1}(0)$}. By convexity of $g$, for all $0 < \epsilon \leq 1$, we have, with $v^\epsilon := \epsilon v_0 + (1-\epsilon) v \in C$,
\begin{align}   \label{eq:local_g_eps_smaller_g_kappa}
  \begin{split}
    g(v^\epsilon) &\leq \epsilon g(v_0) + (1 - \epsilon) g(v) = (1 - \epsilon) g(v) \\
    &< g(v),
  \end{split}
\end{align}
where the equality comes from $g(v_0) = 0$, and the last inequality uses the fact that, because of the assumption $g^{-1}(0) \subseteq f^{-1}(0)$ and $f(v) > \lambda \geq 0$, we must have $g(v) \neq 0$ \hll{--- i.e., taking into account the non-negativity assumption, $g(v) > 0$}. Moreover for small enough $\epsilon$, by continuity of $f$, \hl{the inequality $f(v) > \lambda$ implies that} $f(v^\epsilon) \geq \lambda$.

Therefore, we proved that whenever $f(v) > \lambda$, there exists some $v^\epsilon \in C$ satisfying \hl{both} $f(v^\epsilon) \geq \lambda$ \hl{and} $g(v^\epsilon) < g(v)$\hll{: i.e., $g(v)$ cannot be a minimum of \eqref{eq:def_general_convex_concave_optimisation_problem}}. In other words, for $v$ to achieve the minimum in \eqref{eq:def_general_convex_concave_optimisation_problem}, the condition $f(v) = \lambda$ is necessary \hl{--- which means that} the inequality in \eqref{eq:def_general_convex_concave_optimisation_problem} can be replaced by an equality.
\end{proof}

\subsection{IIB and classic IB}

\hl{We want to impose, in the IIB problem \eqref{eq:def_intertwining_ib}, an additional restriction on $\kappa$ that reduces the latter problem to the Information Bottleneck (IB) problem with source $X$ and relevancy $Y$, i.e., \citep{gilad-bachrachInformationTheoreticTradeoff2003}
\begin{align} \label{eq:ib_problem_primal}
    \argmin_{\substack{q(T_{\ib}|X) \, \in \, C(\mathcal{X}, \mathcal{T}_{\ib}) \; : \\ I_q(Y;T_{\ib}) \geq \lambda}} \; I_q(X;T_{\ib}),
\end{align}
} where $T_{\ib}$ is defined on $\mathcal{T}_{\ib} := \mathbb{N}$, and $I_q(Y;T_{\ib})$ is computed from the marginal $q(Y,T_{\ib})$ of the extension $q(X,Y,T_{\ib})$ of $p(X,Y)$ defined through the Markov chain condition $T_{\ib}-X-Y$, i.e.,  $q(x,y,t_{\ib}) := p(x,y) q(t_{\ib}|x)$. \hl{Let us define the set 
\begin{align*}
  C_{\text{IB}(X,Y)} := \{ \kappa_{\mathcal{X}} \otimes e_{\mathcal{Y}} \hlll{\, : \ \; } \kappa_{\mathcal{X}} \in  C(\mathcal{X}, \mathcal{T}_{\ib}) \} \ \subset \ C(\mathcal{X} \times \mathcal{Y}, \mathcal{T}_{\ib} \times \mathcal{Y}).
\end{align*}
of channels that can compress the $\mathcal{X}$ coordinate but leave the $\mathcal{Y}$ coordinate unchanged. Note that for such channels, the output $T$, defined on $\mathcal{T}_{\ib} \times \mathcal{Y}$,\footnote{\hl{As we defined $\mathcal{T} := \mathbb{N}$, $\mathcal{T}_{\ib} := \mathbb{N}$ and as there is a bijection between $\mathbb{N} \times \mathcal{Y}$ and $\mathbb{N}$, writing here the bottleneck space as $\mathcal{T}_{\ib} \times \mathcal{Y}$ rather than $\mathcal{T}$ is just a difference of presentation.}} can be written $T = (T_{\ib}, Y')$, where $T_{\ib}$ is defined on $\mathcal{T}_{\ib}$ and $Y'$ is a copy of $Y$ --- i.e., $p(T_{\ib}, Y') = p(T_{\ib}, Y)$ and we use the notation $Y'$ instead of $Y$ just because it makes computations clearer below. We now consider the problem
\begin{align} \label{eq:ib_problem_primal_iib_formulation}
    \argmin_{\substack{\kappa \, \in \, C_{\text{IB}(X,Y)} \; : \\ \ D(\kappa(p(X,Y))||\kappa(p(X)p(Y))) = \lambda}} \; I_\kappa(X,Y;T),
\end{align}
which is the IIB problem \eqref{eq:def_intertwining_ib} where we added the constraint that $\kappa$ must be of the form $\kappa_{\mathcal{X}} \otimes e_{\mathcal{Y}}$. It turns out that \eqref{eq:ib_problem_primal_iib_formulation} does coincide with the IB problem, in the following sense:
\begin{proposition} \label{prop:ib_equivalent_to_iib_with_constraint}
    For every $0 \leq \lambda \leq I(X;Y)$, a channel $\kappa_{\mathcal{X}} \otimes e_{\mathcal{Y}} \in C_{\ib}(X,Y)$ solves the problem \eqref{eq:ib_problem_primal_iib_formulation} if and only if $\kappa_{\mathcal{X}} = \kappa_{\mathcal{X}}(T_{\ib}|X)$ solves the IB problem \eqref{eq:ib_problem_primal}.
\end{proposition}
Crucially, note that here $\kappa_{\mathcal{X}} \otimes e_{\mathcal{Y}} \in C_{\ib}(X,Y)$ is entirely determined by $\kappa_{\mathcal{X}}$ through its tensor product with the fixed identity channel $e_{\mathcal{Y}}$, while conversely, $\kappa_{\mathcal{X}}$ is entirely determined by $\kappa_{\mathcal{X}} \otimes e_{\mathcal{Y}}$ through the marginalisation relation
\begin{align*}
    \kappa_{\mathcal{X}}(t_{\ib}|x) = \sum_{y'} \kappa_{\mathcal{X}}(t_{\ib}|x) p(y') = \sum_{y,y'} \kappa_{\mathcal{X}}(t_{\ib}|x) \delta_{y=y'} p(y) = \sum_{y,y'} \kappa_{\mathcal{X}} \otimes e_{\mathcal{Y}}(t_{\ib},y'|x,y)p(y).
\end{align*}
Informally, the only difference between $\kappa_{\mathcal{X}}$ and $\kappa_{\mathcal{X}} \otimes e_{\mathcal{Y}}$ is that $\kappa_{\mathcal{X}} \otimes e_{\mathcal{Y}}$ concatenates the output of $\kappa_{\mathcal{X}}$ with a copy of $Y$. Let us now prove Proposition \ref{prop:ib_equivalent_to_iib_with_constraint}.}
\begin{proof}
    \hl{For $\kappa = \kappa_{\mathcal{X}} \otimes e_{\mathcal{Y}} \in C_{\ib}(X,Y)$, let us write $q(X,Y,T_{\ib},Y')$ the distribution defined by 
    \begin{align}   \label{eq:local_def_q_XYTibYprime}
        q(x,y,t_{\ib},y') := p(x,y) \kappa(t_{\ib},y'|x,y) = p(x,y) \kappa_{\mathcal{X}}(t_{\ib}|x) \delta_{y'=y}.
    \end{align}
    It can be easily verified that then $\kappa(p(X,Y)) = q(T_{\ib},Y)$ and $\kappa(p(X)p(Y)) = q(T_{\ib})p(Y)$, so that 
    \begin{align}   \label{eq:local_dkl_for_ib}
        \begin{split}
            D(\kappa(p(X,Y))||\kappa(p(X)p(Y))) = D(q(T_{\ib},Y)||q(T_{\ib})p(Y))) = I_q(T_{\ib};Y).
        \end{split}
    \end{align}
    On the other hand,
    \begin{align}
        I_\kappa(X,Y;T) &= I_q(X,Y;T_{\ib},Y') \nonumber \\
        &= I_q(X,Y;T_{\ib}) + I_q(X,Y;Y'|T_{\ib}) \label{eq:local_I_XYT_for_ib_1} \\
        &= I_q(X;T_{\ib}) + I_q(Y;Y'|T_{\ib}) + I_q(X;Y'|T_{\ib},Y) \label{eq:local_I_XYT_for_ib_2} \\
        &= I_q(X;T_{\ib}) + H_q(Y|T_{\ib}) \label{eq:local_I_XYT_for_ib_3} \\
        &=  I_q(X;T_{\ib}) - I_q(Y;T_{\ib}) +H(Y), \label{eq:local_I_XYT_for_ib_4}
    \end{align}
    where \hll{line \eqref{eq:local_I_XYT_for_ib_1} uses the chain rule for mutual information, line \eqref{eq:local_I_XYT_for_ib_2} uses the chain rule again and} the fact that from the definition \eqref{eq:local_def_q_XYTibYprime}, under $q$, the Markov chain $T_{\ib} - X - Y$ holds, while line \eqref{eq:local_I_XYT_for_ib_3} uses $I(X;Y'|T_{\ib},Y) = 0$ and $I_q(Y;Y'|T_{\ib}) = H(Y|T_{\ib})$, which are both consequences of $Y'$ being a copy of $Y$. Therefore, combining \eqref{eq:local_dkl_for_ib}, \eqref{eq:local_I_XYT_for_ib_4} and the fact that $H(Y)$ does not depend on $\kappa$, the problem \eqref{eq:ib_problem_primal_iib_formulation} has the same solutions as
    \begin{align} \label{eq:ib_problem_primal_iib_formulation_bis}
        \argmin_{\substack{\kappa \, \in \, C_{\text{IB}(X,Y)} \; : \\ \ I_q(Y;T_{\ib}) = \lambda}} \; \hlcc{\left[ I_q(X;T_{\ib}) - I_q(Y;T_{\ib}) \right]},
    \end{align}
    where $q$ is defined from $\kappa$ through \eqref{eq:local_def_q_XYTibYprime}. But in \eqref{eq:ib_problem_primal_iib_formulation_bis}, as the value of $I_q(Y;T_{\ib})$ is fixed by the constraint, it can be removed from the target function. Moreover, the definition \eqref{eq:local_def_q_XYTibYprime} shows that $\kappa$ is entirely determined by $q(T_{\ib}|X) = \kappa_{\mathcal{X}}$. These two latter facts show that $\kappa$ solves \eqref{eq:ib_problem_primal_iib_formulation_bis} (i.e., solves \eqref{eq:ib_problem_primal_iib_formulation}) if and only if $q(T_{\ib}|X)$ solves 
    \begin{align} \label{eq:local_almostib_problem_primal}
        \argmin_{\substack{q(T_{\ib}|X) \, \in \, C(\mathcal{X}, \mathcal{T}_{\ib}) \; : \\ \ I_q(T_{\ib}; Y) = \lambda}} \; I_q(X;T_{\ib}).
    \end{align}
    Eventually, it can be easily verified that the convex set $C := C(\mathcal{X}, \mathcal{T}_{\ib})$, together with the functions $f(q(T_{\ib}|X)) := I_q(Y;T_{\ib})$ and $g(q(T_{\ib}|X)) := I_q(X;T_{\ib})$, satisfy the assumptions of Lemma \ref{lemma:equality_can_be_inequality_iib}. Thus the equality $I_q(Y;T_{\ib}) = \lambda$ in \eqref{eq:local_almostib_problem_primal} can be replaced by the inequality $I_q(Y;T_{\ib}) \geq \lambda$: in other words, the problem \eqref{eq:local_almostib_problem_primal} can be replaced by the IB problem \eqref{eq:ib_problem_primal}. This ends the proof of the proposition.
    }
\end{proof}

\hll{Let us point out that while the IIB problem is symmetric in $X$ and $Y$, this is not the case for the IB problem, where the source variable and the relevancy variable play different roles. Here, we proved that the IB with source $X$ and relevancy $Y$ can be recovered by adding to \eqref{eq:def_intertwining_ib} the constraint defined by $C_{\ib}(X,Y)$, but similarly, the IB with source $Y$ and relevancy $X$ can be recovered by replacing, in \eqref{eq:ib_problem_primal_iib_formulation}, the set $C_{\ib (X,Y)}$ with the set}
\hl{
\begin{align*}
  C_{\text{IB}(Y,X)} := \{ e_{\mathcal{X}} \otimes \kappa_{\mathcal{Y}} \hlll{\, : \ \; } \kappa_{\mathcal{Y}} \in  C(\mathcal{Y}, \mathcal{T}_{\ib}) \} \ \subset \ C(\mathcal{X} \times \mathcal{Y}, \mathcal{X} \times \mathcal{T}_{\ib}).
\end{align*}
of channels that compress the $\mathcal{Y}$ coordinate but leave the $\mathcal{X}$ coordinate unchanged.}

\subsection{IIB and Symmetric IB}

Let us consider a different restriction on $\kappa$ which will lead to the Symmetric IB \citep{slonimMultivariateInformationBottleneck2006}. \hl{With $\mathcal{T}_\mathcal{X} := \mathbb{N}$ and $\mathcal{T}_\mathcal{Y} := \mathbb{N}$, we define the set
\begin{align*}
  C_{sIB(X,Y)} := \{ \kappa_{\mathcal{X}} \otimes \kappa_{\mathcal{Y}} \hlll{\, : \ \; } \kappa_{\mathcal{X}} \in C(\mathcal{X}, \mathcal{T}_{\mathcal{X}}),  \, \kappa_{\mathcal{Y}} \in C(\mathcal{Y}, \mathcal{T}_\mathcal{Y}) \} \ \subset \ C(\mathcal{X} \times \mathcal{Y}, \mathcal{T}_\mathcal{X} \times \mathcal{T}_\mathcal{Y})
\end{align*}}
of split channels, i.e., of channels that transform $\mathcal{X}$ and $\mathcal{Y}$ separately.\footnote{As we defined $\mathcal{T} := \mathbb{N}$ and as there is a bijection between $\mathbb{N} \times \mathbb{N}$ and $\mathbb{N}$, writing here the \hl{bottleneck} space as $\mathcal{T}_{\mathcal{X}} \times \mathcal{T}_{\mathcal{Y}}$ rather than $\mathcal{T}$ is just a difference of presentation.} \hl{Note that for such channels, the output $T$ can be written $T = (T_X, T_Y)$, where $T_X$ is defined on $\mathcal{T}_{\mathcal{X}}$ and $T_Y$ on $\mathcal{T}_{\mathcal{Y}}$.} We consider the problem
\begin{align} \label{eq:sib_problem_primal_iib_formulation}
  \argmin_{\substack{\kappa \in C_{\text{sIB}(X,Y)} \; : \\ \ D(\kappa(p(X,Y))||\kappa(p(X)p(Y))) = \lambda}} \; I_\kappa(X,Y;T).
\end{align}
We want to show that this problem has the same set of solutions as
\begin{align} \label{eq:sib_problem_primal}
  \argmin_{\substack{q(T_X|X), \, q(T_Y|Y) \; : \\ I_q(T_X;T_Y) \geq \lambda}} \; \hlcc{\left[ I_q(X;T_X) + I_q(Y;T_Y) \right].}
\end{align}

\begin{proposition}
  Let $0 \leq \lambda \leq I(X;Y)$. Then:
  \begin{enumerate}
    \item[$(i)$] In \eqref{eq:sib_problem_primal}, the inequality in the constraint can be replaced by the equality constraint $I_q(T_X;T_Y) = \lambda$.
    \item[$(ii)$] The set of solutions of the problems \eqref{eq:sib_problem_primal_iib_formulation} and \eqref{eq:sib_problem_primal} are identical.
  \end{enumerate}
\end{proposition}

\begin{proof}
    It can be easily verified that the convex set $C := C(\mathcal{X},\mathcal{T}_{\mathcal{X}}) \times C(\mathcal{Y},\mathcal{T}_{\mathcal{Y}})$, together with the functions
  \begin{align*}
    f \, : (q(T_X|X), q(T_Y|Y)) &\mapsto I_q(T_X;T_Y)
  \end{align*}
  and
  \begin{align*}
    g \, : (q(T_X|X), q(T_Y|Y)) &\mapsto I_q(X;T_X) + I(Y;T_Y),
  \end{align*}
  satisfy the assumptions of Lemma \ref{lemma:equality_can_be_inequality_iib}. Thus, the latter proves point $(i)$.

  Let us now prove $(ii)$. \hl{For $\kappa = \kappa_{\mathcal{X}} \otimes \kappa_{\mathcal{Y}}$}, we define the joint distribution $q(X,Y,T_X,T_Y)$ on \hl{$\mathcal{X} \times \mathcal{Y} \times \mathcal{T}_{\mathcal{X}} \times \mathcal{T}_{\mathcal{Y}}$} through
  \begin{align}   \label{eq:local_def_q_sIB}
    q(x,y,t_X,t_Y) := q(x,y) \kappa_{\mathcal{X}}(t_X|x) \kappa_{\mathcal{Y}}(t_Y|y).
  \end{align}
  
  In particular, $q(X,Y,T_X,T_Y)$ is such that the Markov chain $T_X - X - Y - T_Y$ holds. From the latter Markov chain\hll{, using the chain rule for mutual information, we get
  \begin{align}
    I_q(X,Y;T_X, T_Y) &= I(X;T_X, T_Y) + I(Y;T_X, T_Y|X) \nonumber \\
    &= I(X;T_X) + I(X;T_Y|T_X) + I(Y;T_X|X) +I(Y;T_Y|X,T_X) \nonumber \\
    &= I(X;T_X) + I(X;T_Y|T_X) + 0 + I(Y;T_Y|X) \label{eq:loc_1} \\
    &= I(X;T_X) + H(T_Y|T_X) - H(T_Y|T_X,X) + H(T_Y|X) - H(T_Y|X,Y) \label{eq:loc_1bis} \\
    &= I(X;T_X) + H(T_Y|T_X) - H(T_Y|X) + H(T_Y|X) - H(T_Y|Y) \label{eq:loc_2} \\
    &= I(X;T_X) + H(T_Y|T_X) - H(T_Y|Y) \nonumber \\
    &= I(X;T_X) + I(Y;T_Y) - I(T_X; T_Y) \label{eq:local_sIB_is_IIB_with_constraint_1},
  \end{align}
  where line \eqref{eq:loc_1} uses $T_X - X - Y$ and $T_X - X - (T_Y, Y)$; \hlcc{lines \eqref{eq:loc_1bis} and \eqref{eq:local_sIB_is_IIB_with_constraint_1} use the equality $I(A;B) = H(A) - H(A|B)$;} and line \eqref{eq:loc_2} uses $T_X - X - T_Y$ and $X - Y - T_Y$.} Moreover, \hll{it can be verified} that, for $\kappa = \kappa_{\mathcal{X}} \otimes \kappa_{\mathcal{Y}} = q(T_X|T) \otimes q(T_Y|Y)$, \hlll{we have $\kappa(p(X,Y)) = q(T_X,T_Y)$ and $\kappa(p(X)p(Y)) = q(T_X)q(T_Y)$, so that}
  \begin{align}   \label{eq:local_sIB_is_IIB_with_constraint_2}
    D(\kappa(p(X,Y))||\kappa(p(X)p(Y))) \hlll{= D(q(T_X, T_Y) || q(T_X) q(T_Y)) =} I_q(T_X;T_Y).
  \end{align}
  Combining \eqref{eq:local_sIB_is_IIB_with_constraint_1} and \eqref{eq:local_sIB_is_IIB_with_constraint_2} above, we get that the solutions of \eqref{eq:sib_problem_primal_iib_formulation} are also those of
  \begin{align} \label{eq:almost_sib_problem_primal}
    \argmin_{\substack{q(T_X|X), \, q(T_Y|Y) \; : \\ I_q(T_X;T_Y) = \lambda}} \; \hlcc{ \left[ I_q(X;T_X) + I_q(Y;T_Y) - I_q(T_X;T_Y) \right]},
  \end{align}
  But in the latter problem, \hlll{as the value of $I_q(T_X;T_Y)$ is fixed by the constraint, it can be removed from the target function.} Eventually, we can use point $(i)$ to conclude that the solutions \hl{of \eqref{eq:almost_sib_problem_primal}} \hl{coincide with those of the problem \eqref{eq:sib_problem_primal},} which completes the proof.
\end{proof}

Crucially, the problem \eqref{eq:sib_problem_primal} is the Symmetric IB --- more precisely, Ref. \citep{slonimMultivariateInformationBottleneck2006} defines the Lagrangian relaxation \citep{lemarechalLagrangianRelaxation2001} of \eqref{eq:sib_problem_primal}, i.e.,
\begin{align} \label{eq:sib_problem_lagrangian}
  \argmin_{q(T_X|X), \, q(T_Y|Y)} \;\hlcc{ \left[ I_q(X;T_X) + I_q(Y;T_Y) - \beta I_q(T_X;T_Y) \right]},
\end{align}
for varying parameter $\beta \geq 0$. In this sense, the IIB problem \eqref{eq:def_intertwining_ib} with additional constraint of split channel $\kappa = \kappa_{\mathcal{X}} \otimes \kappa_{\mathcal{Y}}$, i.e., the problem \eqref{eq:sib_problem_primal_iib_formulation}, is the Symmetric IB problem.

\section{Proof of Theorem \ref{th:char_equivariance_with_iib_Ixy}}
\label{apd:proof_charac_equivariance_with_iib_solutions}

In most of the proof (Sections \ref{apd:explicit_solution_iib_Ixy} and \ref{apd:charac_equivariance_equivalence_relation}), we will set ourselves in the more general framework of fully supported marginals $p(X)$ and $p(Y)$, but not necessarily fully supported joint distribution $p(X,Y)$. This more general formulation might help for future work to generalise this paper's results. However, at the end the proof (Section \ref{apd:end_proof}) we will use the assumption of fully supported $p(X,Y)$.

\paragraph{Notations} In this proof, we denote channels in $C(\mathcal{X} \times \mathcal{Y}, \mathcal{T})$ by $q(T|X,Y)$ rather than $\kappa$. For $(x,y,t) \in \mathcal{X} \times \mathcal{Y} \times \mathcal{T}$, we write
\begin{align}   \label{eq:def_joint_q_and_qtilde}
  q(x,y,t) := p(x,y) q(t|x,y), \quad \tilde{q}(x,y,t) := p(x)p(y)q(t|x,y),
\end{align}
and $q(T)$, resp. $\tilde{q}(T)$, the corresponding marginals on the bottleneck space $\mathcal{T}$.\footnote{\hlcc{Note the abuse of notation: here, $q(t|x,y)$ is well-defined even when $q(x,y) = p(x,y) = 0$.}} The symbols $\mathcal{S}$ and $\supp(p(X,Y))$ both denote the support of the distribution $p(X,Y)$. For a subset $\mathcal{A}$, we denote by $\mathcal{A}^c$ the complement of $\mathcal{A}$. We consider the equivalence relation
\begin{align}    \label{eq:def_equivalence_relation}
  (x,y) \sim (x',y') \quad \Leftrightarrow \quad \frac{p(x,y)}{p(x)p(y)} = \frac{p(x',y')}{p(x')p(y')},
\end{align}
which is always well-defined, because we assumed that $p(X)$ and $p(Y)$ are fully supported. The \hlcc{equivalence} relation $\sim$ defines a partition of $\mathcal{X} \times \mathcal{Y}$. If $\mathcal{S}^c \neq \emptyset$, then $\mathcal{S}^c$ is an element of this partition, and we write $\{\mathcal{S}_j\}_{j=1,\dots,n}$ for the other elements of the partition, which together thus define a partition of the support $\mathcal{S}$. The latter partition can be seen as the deterministic clustering
\begin{align}   \label{eq:def_clustering_equivalence_rel_support}
  \begin{split}
    \pi_{\mathcal{S}} \, : \quad \mathcal{S} &\longrightarrow \{1, \dots, n\} \\
  (x,y) &\mapsto \sum_{j=1}^n \, j \, \delta_{(x,y) \in \mathcal{S}_j}.
  \end{split}
\end{align}
We also denote by $\pi$ the deterministic clustering defined by the relation $\sim$ on the whole space $\mathcal{X} \times \mathcal{Y}$: explicitly, we set $\pi_{|\mathcal{S}} := \pi_\mathcal{S}$, and if $\mathcal{S}^c \neq \emptyset$, we set $\pi(x,y) = 0$ for $(x,y) \in \mathcal{S}^c$.

\smallskip

As we will see, the clustering $\pi_{\mathcal{S}}$ happens to be the essentially unique solution to \eqref{eq:def_intertwining_ib} for $\lambda=I(X;Y)$. To make this statement precise, we need the following notion \citep{ayInformationGeometry2017}:
\begin{definition}    \label{def:congruent_channel}
  For \hll{finite sets} $\mathcal{A}$ and $\mathcal{B}$, a channel $\gamma$ from $\mathcal{A}$ to $\mathcal{B}$ is called \emph{congruent} if for $a \neq a'$, the supports of $\gamma(B|a)$ and $\gamma(B|a')$ are disjoint. We will denote by $C_{\text{cong}}(\mathcal{A}, \mathcal{B})$ the set of congruent channels from $\mathcal{A}$ to $\mathcal{B}$.
\end{definition}
The definition says that, observing an outcome $b \in \mathcal{B}$ with nonzero probability, one can reconstruct unambiguously the $a \in \mathcal{A}$ which was originally transmitted through the channel. Thus, intuitively, a congruent channel $p(B|A)$ defines a splitting of each symbol $a \in \mathcal{A}$ into the symbol(s) of \hlll{$\text{supp}(p(B|a))$}. Note that permutations of $\mathcal{A}$ are congruent channels with $\mathcal{A}=\mathcal{B}$ and $|\text{supp}(p(B|a))| = 1$ for all $a \in \mathcal{A}$. 

\hlcc{It can be easily verified that $\gamma \in C_{\text{cong}}(\mathcal{A}, \mathcal{B})$ if and only of there is a continuous function $f : \mathcal{B} \rightarrow \mathcal{A}$ such that $f \circ \gamma = e_{\mathcal{A}}$. This can be straightforwadly shown to imply that for a joint distribution $q(A,B) \in \Delta_{\mathcal{A}, \mathcal{B}}$ and a congruent channel $\gamma = \gamma(C|B) \in C_{\text{cong}}(\mathcal{B}, \mathcal{C})$, we get a joint distribution $q(A,B,C) = q(A,B) \gamma(C|B)$ which satisfies $I_q(A;C) = I_q(A;B)$. Intuitively, this means that the composition of a channel $q(B|A)$ at the output by a congruent channel $\gamma$ can be seen as a trivial operation, in that it does not post-process any information.}

\subsection{Explicit form of IIB solutions for $\lambda=I(X;Y)$}
\label{apd:explicit_solution_iib_Ixy}

\begin{theorem}   \label{th:solutions_to_IIB}
  Let $\lambda = I(X;Y)$. The solutions to the IIB problem \eqref{eq:def_intertwining_ib} are the channels of the form
  \begin{align*}
    q(t|x,y) = \begin{cases}
      (\gamma \circ \pi_{\mathcal{S}})(t|x,y) \quad & \text{\emph{if $(x,y) \in \mathcal{S}$}} \\
      q_0(t|x,y) \quad & \text{\emph{if $(x,y) \in \mathcal{S}^c$}}
    \end{cases}
  \end{align*}
  for any congruent channel $\gamma \in C_{\text{cong}}(\{1, \dots,n\}, \mathcal{T})$, and any arbitrary channel $q_0 \in C(\mathcal{S}^c, \mathcal{T})$ on the support's complement.
\end{theorem}

In short, a solution $q(T|X,Y)$ to the IIB for $\lambda = I(X;Y)$ can have an arbitrary effect on the zero probability symbols, but its restriction to the support $\mathcal{S}$ must be, up to permuting or splitting the symbols in $\mathcal{T}$, the clustering $\pi_{\mathcal{S}}$ from \eqref{eq:def_clustering_equivalence_rel_support}. The following corollary is then straightforward:

\begin{corollary} \label{cor:solutions_to_IIB_fullsupport}
  Assume that $p(X,Y)$ is fully supported, and let $\lambda = I(X;Y)$. Then the solutions to the IIB problem \eqref{eq:def_intertwining_ib} are the channels of the form
  \begin{align*}
    q(t|x,y) = (\gamma \circ \pi)(t|x,y)
  \end{align*}
  for any congruent channel $\gamma \in C_{\text{cong}}(\{1, \dots,n\}, \mathcal{T})$, where $\pi$ is the deterministic clustering defined by the relation $\sim$ (see equation \eqref{eq:def_equivalence_relation}).
\end{corollary}

Let us come back to the proof of Theorem \ref{th:solutions_to_IIB}.

\begin{proof}

The following sets, defined for $j=1,\dots,n$, will be central to the proof:
\begin{align} \label{eq:local_def_Tj}
  \mathcal{T}_j^q := \{ t \in \mathcal{T}: \ \exists (x,y) \in \mathcal{S}_j, \ q(t|x,y) > 0 \}.
\end{align}
Intuitively, $\mathcal{T}_j^q$ is the ``probabilistic image set'' of $\mathcal{S}_j$ through $q(T|X,Y)$: i.e., it is the subset of $\mathcal{T}$ that can be achieved with nonzero probability starting from inputs $(x,y)$ in $\mathcal{S}_j$ and using the channel $q(T|X,Y)$. \hlcc{Most of the proof below consists, intuitively, in proving that each $\mathcal{T}_j^q$ is ``essentially'' a single bottleneck symbol --- i.e., up to the trivial operation of permuting or splitting symbols with a congruent channel.} It will also be useful to consider, for $t \in \mathcal{T}$,
\begin{align} \label{eq:local_def_Sqt}
  \mathcal{S}^q_t := \{ (x,y) \in \mathcal{S} : \ q(t|x,y) > 0 \},
\end{align}
which can be seen as the ``probabilistic pre-image set'' of $t$ through $q(T|X,Y)$. 

Note that the constraint function in the IIB problem \eqref{eq:def_intertwining_ib} can be rewritten $D(q(T)||\tilde{q}(T))$. The following lemma shows that \hlc{the constraint} $D(q(T)||\tilde{q}(T)) = I(X;Y)$ is characterised by the fact that $\frac{p(x,y)}{p(x)p(y)}$ is constant on the ``pre-image'' $\mathcal{S}^q_t$ of every symbol $t$:

\begin{lemma} \label{lemma:pre-charac_constraint_set_ib}
  Let $q(T|X,Y) \in C(\mathcal{X} \times \mathcal{Y}, \mathcal{T})$. Then we always have $D(q(T)||\tilde{q}(T)) \leq I(X;Y)$, and $D(q(T)||\tilde{q}(T)) = I(X;Y)$ if and only if, for all $t \in \mathcal{T}$, there exists some $\mathcal{S}_j$ such that
  \begin{align}   \label{eq:local_s_q_t_in_s_j}
    \mathcal{S}^q_t \subseteq \mathcal{S}_j.
  \end{align}
\end{lemma}

\begin{proof}
  We have
\begin{align*}
  D(q(T)||\tilde{q}(T)) &= \sum_{t} \left(\sum_{x,y} q(t|x,y) p(x,y) \right) \log\left( \frac{\sum_{x,y} q(t|x,y) p(x,y)}{\sum_{x,y} q(t|x,y) p(x)p(y)} \right) \\
  &= \sum_{t} \left(\sum_{(x,y) \in \mathcal{S}} q(t|x,y) p(x,y) \right) \log\left( \frac{\sum_{(x,y) \in \mathcal{S}} q(t|x,y) p(x,y)}{\sum_{(x,y) \in \mathcal{S}} q(t|x,y) p(x)p(y)} \right),
\end{align*}
while
\begin{align*}
  I(X;Y) &= \sum_{x,y} p(x,y) \log \left( \frac{p(x,y)}{p(x)p(y)} \right) \\
  &= \sum_{(x,y) \in \mathcal{S}} p(x,y) \log \left( \frac{p(x,y)}{p(x)p(y)} \right) \\
  &= \sum_{(x,y) \in \mathcal{S}} \left( \sum_t q(t|x,y) p(x,y) \right) \log\left( \frac{q(t|x,y) p(x,y)}{q(t|x,y) p(x)p(y)} \right) \\
  &= \sum_t \sum_{(x,y) \in \mathcal{S}} q(t|x,y) p(x,y) \log\left( \frac{q(t|x,y) p(x,y)}{q(t|x,y) p(x)p(y)} \right),
  \end{align*}
where we use the convention $0\log(\frac{0}{0}) = 0$. But from the log-sum inequality, for all $t \in \mathcal{T}$,
\begin{align}   \label{eq:local_inequality_logsum}
  \begin{split}
    \left(\sum_{(x,y) \in \mathcal{S}} q(t|x,y) p(x,y) \right) \log & \left( \frac{\sum_{(x,y) \in \mathcal{S}} q(t|x,y) p(x,y)}{\sum_{(x,y) \in \mathcal{S}} q(t|x,y) p(x)p(y)} \right) \\
    & \leq \sum_{(x,y) \in \mathcal{S}} q(t|x,y) p(x,y) \log\left( \frac{q(t|x,y) p(x,y)}{q(t|x,y) p(x)p(y)} \right).
  \end{split}
\end{align}
So that $D(q(T)||\tilde{q}(T)) \leq I(X;Y)$, \del with equality if and only if, for all $t \in \mathcal{T}$, it holds in \eqref{eq:local_inequality_logsum}. From the equality case of the log-sum inequality \citep{csiszarInformationTheoryCoding2011}, the latter is equivalent to the existence of nonzero constants $(\alpha_t)_{t \in \mathcal{T}}$ such that
\begin{align*}
  \forall (x,y) \in \mathcal{S}, \quad \quad q(t|x,y) p(x,y) \, = \, \alpha_t \, q(t|x,y) p(x)p(y),
\end{align*}
i.e., such that, for every $t$, the quantity $\frac{p(x,y)}{p(x)p(y)}$ is constant on the subset of elements $(x,y)$ \hlcc{for which} $q(t|x,y)>0$. Recalling the definitions \eqref{eq:local_def_Sqt} of $\mathcal{S}^q_t$ and \eqref{eq:def_equivalence_relation} of the relation $\sim$ defining the sets $\mathcal{S}_j$, we thus proved the following: we have $D(q(T)||\tilde{q}(T)) = I(X;Y)$ if and only if, for all $t \in \mathcal{T}$, there exists some $\mathcal{S}_j$ such that
\begin{align}   
  \mathcal{S}^q_t \, \subseteq \, \mathcal{S}_j.
\end{align}
\end{proof}

To state the next lemma, we define, for a given $q(T|X,Y)$, the channel $\gamma_q \in C(\{ 1,\dots,n \}, \mathcal{T})$ through
\begin{align}   \label{eq:def:gamma_q}
  \gamma_q(t|j) := q(t|\mathcal{T}_j^q) = \frac{q(t)}{q(\mathcal{T}_j^q)} \delta_{t \in \mathcal{T}_j^q}.
\end{align}
Note that \hlcc{here the indices $j = 1, \dots, n$ are thought of as indexing the elements $\mathcal{S} _j$ of the partition of $\mathcal{S} := \supp(p(X,Y)) \subseteq \mathcal{X} \times \mathcal{Y}$; and that} the support of $\gamma_q(\cdot|j)$ is \hlcc{exactly the ``probabilistic image'' $\mathcal{T}_j^q$ of $\mathcal{S}_j$ through $q(T|X,Y)$}.

\begin{lemma}   \label{lemma:in_delta_iib_iff_sends_partition_to_partition}
  Let $q(T|X,Y) \in C(\mathcal{X} \times \mathcal{Y}, \mathcal{T})$. Then the following are equivalent:
  \begin{enumerate}
    \item[$(i)$] $D(q(T)||\tilde{q}(T)) = I(X;Y)$,
    \item[$(ii)$] $\{ \mathcal{T}_j^q\}_{j=1,\dots,n}$ is a partition of $\supp(q(T)) \subseteq \mathcal{T}$,
    \hlc{\item[$(iii)$] The channel $\gamma_q$ defined in \eqref{eq:def:gamma_q} is congruent.}
  \end{enumerate}
\end{lemma}

\begin{proof}
Note that it clearly follows from the definition \eqref{eq:local_def_Tj} that the union of the sets $\mathcal{T}_j^q$ is $\supp(q(T))$, so that these sets define a partition of $\supp(q(T))$ if and only if they are disjoint. Moreover, the definition \eqref{eq:local_def_Tj} of $\mathcal{T}_j^q$ can be reformulated as
\begin{align}   \label{eq:local_def_Tj_alternative}
  \mathcal{T}_j^q= \{ t \in \supp(q(T)) : \ \mathcal{S}^q_t \cap \mathcal{S}_j \neq \emptyset\}.
\end{align}
which means, intuitively, that a symbol $t$ is in the \hlll{(probabilistic)} image $\mathcal{T}_j^q$ of $\mathcal{S}_j$ through $q(T|X,Y)$ if and only if the \hlll{(probabilistic)} pre-image $S^q_t$ of $t$ intersects the set $\mathcal{S}_j$.

Assume that $D(q(T)||\tilde{q}(T)) = I(X;Y)$ holds. Then Lemma \ref{lemma:pre-charac_constraint_set_ib} and the fact that the $\mathcal{S}_j$ are disjoint imply that $\mathcal{S}^q_t \cap \mathcal{S}_j \neq \emptyset \Leftrightarrow \mathcal{S}^q_t \subseteq \mathcal{S}_j$ \hlcc{for $t \in \supp(q(T))$ (note that $q(t)>0$ implies $\mathcal{S}_t^q \neq \emptyset$).} So that
\begin{align}   \label{eq:local_charac_tj_with_S_q_t}
  \mathcal{T}_j^q= \{ t \in \supp(q(T)) : \ \mathcal{S}^q_t \subseteq \mathcal{S}_j \}.
\end{align}
Therefore, once again because the $\mathcal{S}_j$ are disjoint, the sets $\mathcal{T}_j^q$ must also be disjoint, and they define a partition of $\supp(q(T))$.

Conversely, assume that $\{ \mathcal{T}_j^q \}_{j=1,\dots,n}$ is a partition of $\supp(q(T))$. If there is some $t \in \supp(q(T))$ such that we have both $\mathcal{S}^q_t \cap \mathcal{S}_j \neq \emptyset$ and $\mathcal{S}^q_t \cap \mathcal{S}_{j'} \neq \emptyset$, \del then from \eqref{eq:local_def_Tj_alternative}, we have $t \in \mathcal{T}_j^q \cap \mathcal{T}_{j'}^q$. Thus for all $t \in \supp(q(T))$, there is \hlc{at most one} $j \in \{1, \dotsm,n\}$ such that $\mathcal{S}^q_t \cap \mathcal{S}_j \neq \emptyset$. As the union of the $\mathcal{S}_j$ over $j \in \{ 1, \dots, n \}$ is $\mathcal{S}$, and as by definition, $\mathcal{S}^q_t$ is included in $\mathcal{S}$, this means that \hlc{there exists a (unique) $j$ such that} $\mathcal{S}_t^q \subseteq \mathcal{S}_j$. Therefore, from Lemma \ref{lemma:pre-charac_constraint_set_ib}, we must have  $D(q(T)||\tilde{q}(T)) = I(X;Y)$, and the equivalence of points $(i)$ and $(ii)$ is proven.

\hlc{The equivalence of points $(ii)$ and $(iii)$ is a consequence of the fact that for all $j$, the support of $\gamma_q(\cdot|j)$ is precisely $\mathcal{T}_j^q$. Thus $\gamma_q$ is congruent if and only if the $\mathcal{T}_j^q$ are disjoint, which as already noticed, is equivalent to them defining a partition of $\supp(q(T))$.}

\end{proof}

\hlc{Now that we described what it means for a channel $q(T|X,Y)$ to satisfy the constraint of the IIB problem, let us describe the implications of it also minimising the target function. For that purpose, it will be convenient to consider, for any $q(T|X,Y) \in C(\mathcal{X} \times \mathcal{Y}, \mathcal{T})$ satisfying the constraint $D(q(T)||\tq(T)) = D(p(X,Y)||\hlcc{p(X)p(Y)})$, the channel $q'(T|X,Y)$ defined by
\begin{align} \label{eq:loc:def_qprime_1}
  q'(t|x,y) := \begin{cases}
    q(t|\mathcal{T}_j^q) q(\mathcal{T}_j^q|x,y) \quad &\text{if $(x,y) \in \mathcal{S}$} \\
    q(t|x,y) \quad &\text{if $(x,y) \in \mathcal{S}^c$}
  \end{cases},
\end{align}
where $j$ is the unique index such that $t \in \mathcal{T}_j^q$. We know that such a $j$ exists because, from Lemma \ref{lemma:in_delta_iib_iff_sends_partition_to_partition} and the assumption that $q(T|X,Y)$ satisfies the constraint, \hlcc{$\{ \mathcal{T}_j^q \}_j$ is a partition of $\supp(q(T))$. The latter also ensures that $q'(T|X,Y)$ thus defined is indeed a conditional probability. Intuitively, the channel $q'(T|X,Y)$ modifies $q(T|X,Y)$ so that it becomes factorisable, on $\mathcal{S}$, by the clustering $\pi_{\mathcal{S}}$ (see equation \eqref{eq:def_clustering_equivalence_rel_support} and point $(i)$ in Lemma \ref{lemma:charac_q_equals_qprime} below).} We also consider the "probabilistic images" of each $\mathcal{S}_j$ through $q'$, i.e.,
\begin{align*}
  \mathcal{T}_j^{q'} = \{ t \in \mathcal{T}: \ \exists (x,y) \in \mathcal{S}_j, \ \ q'(t|x,y) > 0 \}.
\end{align*}
\begin{lemma} \label{lemma:charac_q_equals_qprime}
  Let $q(T|X,Y) \in C(\mathcal{X} \times \mathcal{Y}, \mathcal{T})$ such that $D(q(T)||\tq(T)) = D(p(X,Y)||\hlcc{p(X)p(Y)})$. Then:
  \begin{enumerate}
    \item[$(i)$] For $(x,y) \in \mathcal{S}$, we have $q(\mathcal{T}_j^q|x,y) = \delta_{(x,y) \in \mathcal{S}_j}$.
    \item [$(ii)$]$q'(T) = q(T)$.
    \item[$(iii)$] $\mathcal{T}_j^q = \mathcal{T}_j^{q'}$ for all $j$.
    \item[$(iv)$] $I_q(X,Y;T) \geq I_{q'}(X,Y;T)$, and equality holds if and only if $q(T|X,Y) = q'(T|X,Y)$.
  \end{enumerate}
\end{lemma}}

\begin{proof}
  \hlc{First, let us recall that because we assume $D(q(T)||\tq(T)) = D(p(X,Y)||\hlcc{p(X)p(Y)})$, Lemma \ref{lemma:in_delta_iib_iff_sends_partition_to_partition} ensures that $\{ \mathcal{T}_j^q \}_j$ is a partition of $\supp(q(T))$.}

  \hlc{$(i)$. If $(x,y) \in \mathcal{S}_j$, then by definition of $\mathcal{T}_j^q$ as the probabilistic image set of $\mathcal{S}_j$, we have $q(\mathcal{T}_j^q|x,y) = 1$. If $(x,y) \in \mathcal{S} \setminus \mathcal{S}_j$, then $q(\mathcal{T}_j^q|x,y) = 0$ is a consequence of the fact that the $\{ \mathcal{T}_{j'}^q \}_{j'}$ are disjoint.}
  
  \hlc{$(ii)$. For $t \in \supp(q(T))$ and $j$  the unique index such that $t \in \mathcal{T}_j^q$,
  \begin{align*}
      q'(t) &= \sum_{(x,y) \in \mathcal{X}} q'(t|x,y) p(x,y) \\
      &= \sum_{(x,y) \in \mathcal{S}} q'(t|x,y) p(x,y) \\
      &= \sum_{(x,y) \in \mathcal{S}}  q(t|\mathcal{T}_j^q) q(\mathcal{T}_j^q|x) p(x,y) \\
      &= q(t|\mathcal{T}_j^q) q(\mathcal{T}_j^q) \\
      &= q(t),
  \end{align*}
  \hlcc{where the last line uses $q(t|\mathcal{T}_j^q) = \frac{q(t)}{q(\mathcal{T}_j^q)} \delta_{t \in \mathcal{T}_j^q}$.}
  }

  \hlc{$(iii)$. For fixed index $j$ and $t \in \mathcal{T}$,
  \begin{align*}
      \exists (x,y) \in \mathcal{S}_j: \ q'(t|x,y) >0 \quad &\Leftrightarrow \quad \exists (x,y) \in \mathcal{S}_j: \ q(t|\mathcal{T}_j^q) q(\mathcal{T}_j^q|x,y) >0 \\
      &\Leftrightarrow \quad \exists (x,y) \in \mathcal{S}_j: \ \text{$t \in \mathcal{T}_j^q$ and $(x,y) \in \mathcal{S}_j$} \\
      &\Leftrightarrow \quad t \in \mathcal{T}_j^q
  \end{align*}
  where the first line uses $\mathcal{S}_j \subseteq \mathcal{S}$; and the second line uses point $(i)$ and $q(t|\mathcal{T}_j^q) = \frac{q(t)}{q(\mathcal{T}_j^q)} \delta_{t \in \mathcal{T}_j^q}$.}
  
  \hlc{$(iv)$. First write, with the convention $0 \log(\frac{0}{0}) = 0$,
  \begin{align} \label{eq:loc:rewriting_mutual_info}
    \begin{split}
    I_q(X,Y;T) &= \sum_{x,y,t} p(x,y) q(t|x,y) \log\left( \frac{q(t|x,y)}{q(t)} \right) \\
    &= \sum_{(x,y) \in \mathcal{S}} p(x,y) \sum_{t \in \supp(q(T))} q(t|x,y) \log\left( \frac{q(t|x,y)}{q(t)} \right) \\
    &= \sum_{j=1}^n \sum_{(x,y) \in \mathcal{S}_j} p(x,y) \sum_{t \in \supp(q(T))} q(t|x,y) \log\left( \frac{q(t|x,y)}{q(t)} \right) \\
    &= \sum_{j=1}^n \sum_{(x,y) \in \mathcal{S}_j} p(x,y) \sum_{t \in \mathcal{T}_j^q} q(t|x,y) \log\left( \frac{q(t|x,y)}{q(t)} \right),
  \end{split}
  \end{align}
  where the second equality uses the fact that if $(x,y) \in \mathcal{S}$, then $q(t) = 0$ implies that $q(t|x) = 0$; and the last equality follows from the definition of $T^q_j$ as the ``probabilistic image set'' of $\mathcal{S}_j$ (see equation \eqref{eq:local_def_Tj}).
  Yet, using once again the log-sum inequality, we have, for all $j=1,\dots,n$ and all $(x,y) \in \mathcal{S}_j$,
  \begin{align*}
    \sum_{t \in \mathcal{T}_j^q} q(t|x,y) \log & \left( \frac{q(t|x,y)}{q(t)} \right) \geq  \left(\sum_{t \in \mathcal{T}_j^q} q(t|x,y) \right) \log \left( \frac{\sum_{t \in \mathcal{T}_j^q} q(t|x,y)}{\sum_{t \in \mathcal{T}_j^q} q(t)} \right),
  \end{align*}
  i.e,
  \begin{align}   \label{eq:local_logsum2}
    \sum_{t \in \mathcal{T}_j^q} q(t|x,y) \log & \left( \frac{q(t|x,y)}{q(t)} \right) \geq  q(\mathcal{T}_j^q|x,y) \log \left( \frac{q(\mathcal{T}_j^q|x,y)}{q(\mathcal{T}_j^q)} \right),
  \end{align}
  with equality if and only if for all $t \in \mathcal{T}_j^q$,
  \begin{align*}
    \frac{q(t|x,y)}{q(t)} = \frac{q(\mathcal{T}_j^q|x,y)}{q(\mathcal{T}_j^q)},
  \end{align*}
  i.e.,
  \begin{align} \label{eq:local_charac_minimal_IXYT_prov}
    q(t|x,y) = q(t|\mathcal{T}_j^q) q(\mathcal{T}_j^q|x,y).
  \end{align}
  Moreover, note that for $(x,y) \in \mathcal{S}_j \subseteq \mathcal{S}$, the right-hand-side of \eqref{eq:local_logsum2} can be rewritten
  \begin{align*}
      q(\mathcal{T}_j^q|x,y) \log \left( \frac{q(\mathcal{T}_j^q|x,y)}{q(\mathcal{T}_j^q)} \right) &= \left( \sum_{t \in \mathcal{T}_j^{q}} q(t|\mathcal{T}_j^q) \right) q(\mathcal{T}_j^q|x,y) \log \left( \frac{q(\mathcal{T}_j^q|x,y)}{q(\mathcal{T}_j^q)} \right) \\
      &= \sum_{t \in \mathcal{T}_j^{q}} q(t|\mathcal{T}_j^q) q(\mathcal{T}_j^q|x,y) \log \left( \frac{q(t|\mathcal{T}_j^q) q(\mathcal{T}_j^q|x,y)}{q(t)} \right) \\
      &= \hlcc{\sum_{t \in \mathcal{T}_j^{q}} q'(t|x,y) \log \left( \frac{q'(t|x,y)}{q(t)} \right)} \\
      &= \hlcc{\sum_{t \in \mathcal{T}_j^{q'}} q'(t|x,y) \log \left( \frac{q'(t|x,y)}{q'(t)} \right)},
  \end{align*}
  where the last equality uses points $(ii)$ and $(iii)$ just proven. Thus, multipling by $p(x,y)$ and summing both sides of \eqref{eq:local_logsum2} over $j=1,\dots, n$ and $(x,y) \in \mathcal{S}_j$, we get $I_q(X,Y;T) \geq I_{q'}(X,Y;T)$. Considering the equality case of the log-sum inequality then yields, from \eqref{eq:local_charac_minimal_IXYT_prov},
  \begin{align*}
    I_q(X,Y;T) = I_{q'}(X,Y;T) \quad &\Leftrightarrow \quad \forall j=1,\dots,n, \forall (x,y) \in \mathcal{S}_j, \forall t \in \mathcal{T}_j^q, \quad q(t|x,y) = q(t|\mathcal{T}_j^q) q(\mathcal{T}_j^q|x,y) \\
    &\Leftrightarrow \quad \forall (x,y) \in \mathcal{S}, \forall j=1,\dots,n, \forall t \in \mathcal{T}_j^q, \quad q(t|x,y) = q(t|\mathcal{T}_j^q) q(\mathcal{T}_j^q|x,y) \\
    &\Leftrightarrow \quad \forall (x,y) \in \mathcal{S}, \forall t \in \supp(q(T)), \quad q(t|x,y) = q(t|\mathcal{T}_j^q) q(\mathcal{T}_j^q|x,y) \\
    &\Leftrightarrow \quad \forall (x,y) \in \mathcal{S}, \forall t \in \mathcal{T}, \quad q(t|x,y) = q(t|\mathcal{T}_j^q) q(\mathcal{T}_j^q|x,y) \\
    &\Leftrightarrow \quad q(T|X,Y) = q'(T|X,Y),
  \end{align*}
  where the second line uses point $(i)$ and the fact that $q(t|x,y) = 0$ for $t \in \mathcal{T}_j^q$ but $(x,y) \in \mathcal{S} \setminus \mathcal{S}_j$ (because the $\{ \mathcal{T}_{j'}^q \}_{j'}$ are disjoint); the third one that $\cup_j \mathcal{T}_j^q = \supp(q(T))$; the fourth one that $(x,y) \in \mathcal{S} := \supp(p(X,Y))$ and $t \notin \supp(q(T))$ implies $q(t|x,y) = q(t|\mathcal{T}_j^q) = 0$.}
\end{proof}

\hlc{\begin{lemma}   \label{lemma:almost_charc_iib_Ixy}
  Let $q(T|X,Y) \in C(\mathcal{X} \times \mathcal{Y}, \mathcal{T})$. If $q(T|X,Y)$ solves the IIB problem \eqref{eq:def_intertwining_ib} with $\lambda = D(p(X,Y)||\hlcc{p(X)p(Y)})$, then for all $(x,y) \in \mathcal{S}$,
  \begin{align} \label{eq:pre-charac_sol_to_IIB_Ixy}
      q(t|x,y) = \sum_{j=1}^n q(t|\mathcal{T}_{j}^q) \delta_{(x,y) \in \mathcal{S}_j}.
  \end{align}
\end{lemma}}

\begin{proof}
  \hlc{Let us fix a solution $q(T|X,Y)$ to the IIB problem with $\lambda = D(p(X,Y)||\hlcc{p(X)p(Y)})$. In particular, $q$ satisfies the constraint $D(q(T)||\tq(T)) = D(p(X,Y)||\hlcc{p(X)p(Y)})$, so that from Lemma \ref{lemma:in_delta_iib_iff_sends_partition_to_partition}, $\{ \mathcal{T}_j^{q} \}_j$ is a partition of $\supp(q(T))$. Thus from points $(ii)$ and $(iii)$ in Lemma \ref{lemma:charac_q_equals_qprime}, $\{ \mathcal{T}_j^{q'} \}_j$  is a partition of $\supp(q'(T))$. From Lemma  \ref{lemma:in_delta_iib_iff_sends_partition_to_partition} again, we conclude that $D(q'(T)||\tilde{q'}(T)) = D(p(X,Y)||\hlcc{p(X)p(Y)})$: i.e., $q'(T|X,Y)$ satifies the constraint of the IIB problem.}

  \hlc{On the other hand, from point $(iv)$ in Lemma \ref{lemma:charac_q_equals_qprime}, $q(T|X,Y) \neq q'(T|X,Y)$ is only possible if $I_q(X,Y;T) > I_{q'}(X,Y;T)$. Thus, if $q(T|X,Y) \neq q'(T|X,Y)$, then $q'(T|X,Y)$ both satisfies the constraint of the IIB problem and yields a smaller target function than $q(T|X,Y)$, which is incompatible with $q(T|X,Y)$ solving the IIB problem. In other words, we must have $q(T|X,Y) = q'(T|X,Y)$: i.e., for all $(x,y) \in \mathcal{S}$, we have $q(t|x,y) = q(t|\mathcal{T}_j^q) q(\mathcal{T}_j^q|x,y)$. We conclude with point $(i)$ in Lemma \ref{lemma:charac_q_equals_qprime}, and the fact that $\{ \mathcal{S}_j\}_j$ is a partition of $\mathcal{S}$.}
\end{proof}

Now let $q(T|X,Y)$ be a channel that solves the IIB problem \eqref{eq:def_intertwining_ib} with $\lambda = D(p(X,Y)||\hlcc{p(X)p(Y)})$. The conclusion of Lemma \ref{lemma:almost_charc_iib_Ixy} can be reformulated as the assertion that for all $(x,y) \in \mathcal{S}$, $t \in \mathcal{T}$,
\begin{align*}
    q(t|x,y) = (\gamma_q \circ \pi_{\mathcal{S}})(t|x,y),
\end{align*}
where we recall that $\pi_{\mathcal{S}}$ is the deterministic clustering defined by the partition $\{ \mathcal{S}_j \}_j$ of $\mathcal{S}$ (see \eqref{eq:def_clustering_equivalence_rel_support}), and $\gamma_q$ is defined in \eqref{eq:def:gamma_q}. Moreover, Lemma \ref{lemma:in_delta_iib_iff_sends_partition_to_partition} ensures that $\gamma_q$ is congruent. Therefore, we have proven that any solution to the IIB problem \eqref{eq:def_intertwining_ib} for $\lambda=D(p(X,Y)||\hlcc{p(X)p(Y)})$ must be of the form
\begin{align} \label{eq:local:form_solution_set_IIB}
    q(t|x,y) = \begin{cases}
    (\gamma \circ \pi_{\mathcal{S}})(t|x,y) \quad & \text{\emph{if $(x,y) \in \mathcal{S}$}} \\
    q_0(t|x,y) \quad & \text{\emph{if $(x,y) \in \mathcal{S}^c$}}
    \end{cases}
\end{align}
for some congruent channel $\gamma \in C_{\text{cong}}(\{1, \dots,n\}, \mathcal{T})$, and some arbitrary channel $q_0 \in C(\mathcal{S}^c, \mathcal{T})$ on the support's complement. I.e., \hlcc{denoting} $E$ the set of channels of the latter form \eqref{eq:local:form_solution_set_IIB}, we proved that the set of solutions to the IIB problem \eqref{eq:def_intertwining_ib} for $\lambda=D(p(X,Y)||\hlcc{p(X)p(Y)})$ is included in $E$.

\hlc{To end the proof of Theorem \ref{th:solutions_to_IIB}, let us prove that $E$ is included in the solutions to the IIB for $\lambda = I(X;Y)$.
\begin{lemma} \label{lemma:value_target_constraint_for_solutions}
  For $q(T|X,Y) \in E$, we have 
  \begin{align*}
    D(q(T)||\tilde{q}(T)) = I(X;Y)
  \end{align*} 
  and
  \begin{align*}
    I_{q}(X,Y;T) = H(\pi_\mathcal{S}(X,Y)).
  \end{align*}
\end{lemma}
}
\begin{proof}
\hlc{Let $q(T|X,Y) \in E$. \hlcc{Then $\gamma$ from the definition \eqref{eq:local:form_solution_set_IIB} of $q(T|X,Y)$ coincides with the channel $\gamma_q$ defined in \eqref{eq:def:gamma_q}. Indeed, let us fix $j$. First, we have
\begin{align}   \label{eq:local:supp_gamma_is_Tjq}
    \supp(\gamma(\cdot|j)) = \Tcal_j^q,
\end{align}
because for $t \in \Tcal$,
\begin{align*}
    \gamma(t|j) > 0 \quad &\Leftrightarrow \quad \exists (x,y) \in \Scal_j, \ \ \gamma(t|j) \delta_{(x,y) \in \Scal_{j}} > 0 \\
    &\Leftrightarrow \quad \exists (x,y) \in \Scal_j, \ \ \sum_{j'=1}^n \gamma(t|j') \delta_{(x,y) \in \Scal_{j'}} > 0 \\
    &\Leftrightarrow \quad \exists (x,y) \in \Scal_j, \ \ (\gamma \circ \pi_{\Scal}) (t|x,y) > 0 \\
    &\Leftrightarrow \quad \exists (x,y) \in \Scal_j, \ \ q(t|x,y) > 0 \\
    &\Leftrightarrow \quad t \in \Tcal_j^q,
\end{align*}
where the first lines uses $\Scal_j \neq \emptyset$, and the last one uses the definition \eqref{eq:local_def_Tj} of $\Tcal_j^q$. Thus, $t \notin \Tcal_j^q$ implies $\gamma(t|j) = 0$, and for $t \in \Tcal_j^q$,
\begin{align*}
    q(t) &= \sum_{(x,y) \in \Xcal \times \Ycal} q(t|x,y) p(x,y) \\
    &= \sum_{(x,y) \in \Scal_j} q(t|x,y) p(x,y) \\
    &= \sum_{(x,y) \in \Scal_j} \gamma(t|j) p(x,y) \\
    &= \gamma(t|j) p(\Scal_j),
\end{align*}
where the second line uses the definition \eqref{eq:local_def_Tj} of $\Tcal_j^q$, and the third line uses the definition \eqref{eq:local:form_solution_set_IIB} of $q(T|X,Y)$. As a consequence, we also have
\begin{align}   \label{qTj_is_pSj}
    q(\Tcal_j^q) &= \sum_{t \in \Tcal_j^q} q(t) = p(\Scal_j) \sum_{t \in \Tcal_j^q} \gamma(t|j) = p(\Scal_j),
\end{align}
where the last equality uses \eqref{eq:local:supp_gamma_is_Tjq}. Thus we do have, for all $t \in \Tcal$,
\begin{align}   \label{eq:loc:gamma_is_gammaq}
    \gamma(t|j) = \frac{q(t)}{q(\Tcal_j^q)} \delta_{t \in \Tcal_j^q},
\end{align}
i.e., $\gamma(t|j)=\gamma_q(t|j)$ (see equation \eqref{eq:def:gamma_q}).} As $\gamma$ is assumed congruent, Lemma \ref{lemma:in_delta_iib_iff_sends_partition_to_partition} then implies that $D(q(T)||\tilde{q}(T)) = I(X;Y)$.}

On the other hand, as in \eqref{eq:loc:rewriting_mutual_info}, we can write \hlcc{
\begin{align} 
  I_q(X,Y;T) &= \sum_{j=1}^n \sum_{(x,y) \in \mathcal{S}_j} p(x,y) \sum_{t \in \mathcal{T}_j^q} q(t|x,y) \log\left( \frac{q(t|x,y)}{q(t)} \right) \notag \\
  &= \sum_{j=1}^n \sum_{(x,y) \in \mathcal{S}_j} p(x,y) \sum_{t \in \mathcal{T}_j^q} q(t|x,y) \log\left( \frac{\gamma(t|j)}{q(t)} \right) \label{eq:loc:I_xy_on_supports_1} \\
  &= \sum_{j=1}^n \sum_{(x,y) \in \mathcal{S}_j} p(x,y) \sum_{t \in \mathcal{T}_j^q} q(t|x,y) \log\left( \frac{q(t)}{q(t)q(\Tcal_j^q)} \right) \label{eq:loc:I_xy_on_supports_2} \\
  &= \sum_{j=1}^n \log\left( \frac{1}{q(\Tcal_j^q)} \right)  \sum_{(x,y) \in \mathcal{S}_j} p(x,y) \sum_{t \in \mathcal{T}_j^q} q(t|x,y) \notag \\
  &= \sum_{j=1}^n \log\left( \frac{1}{q(\Tcal_j^q)} \right)  p(\Scal_j) \label{eq:loc:I_xy_on_supports_3} \\
  &= \sum_{j=1}^n p(\Scal_j) \log\left( \frac{1}{p(\Scal_j)} \right) \label{eq:loc:I_xy_on_supports_4} \\
  &= H(\pi_\mathcal{S}(X,Y)), \notag
\end{align}
where line \eqref{eq:loc:I_xy_on_supports_1} uses the definition \eqref{eq:local:form_solution_set_IIB} of $q(T|X,Y)$; line \eqref{eq:loc:I_xy_on_supports_2} uses equation \eqref{eq:loc:gamma_is_gammaq}; line \eqref{eq:loc:I_xy_on_supports_3} uses the definition \eqref{eq:local_def_Tj} of $\Tcal_j^q$; and line \eqref{eq:loc:I_xy_on_supports_4} uses equation \eqref{qTj_is_pSj}.} \del
\end{proof}
\hlc{Now, because the IIB problem is defined as the minimisation of a continuous function on a compact domain, it has at least one solution, say $q_\ast(T|X,Y)$, which we know belongs to $E$ \hlcc{(we already proved that any solution to the IIB with $\lambda=D(p(X,Y)||p(X)p(Y))$ belongs to $E$)}. But Lemma \ref{lemma:value_target_constraint_for_solutions} then implies that for all $q(T|X,Y) \in E$, we have $D(q(T)||\tilde{q}(T)) = D(q_\ast(T)||\tilde{q_\ast}(T))$ and $I_{q}(X,Y;T) = I_{q_\ast}(X,Y;T)$. Thus any $q(T|X,Y) \in E$ must also be a solution.}

\end{proof}

\subsection{Characterisation of equivariances with the equivalence relation}
\label{apd:charac_equivariance_equivalence_relation}

In this part, we characterise the equivariance group of $p(Y|X)$ with the equivalence relation $\sim$ (see equation \eqref{eq:def_equivalence_relation}), thanks to the specific assumption that $p(Y)$ is uniform (see Theorem \ref{th:char_equivariance_with_iib_Ixy}).

\begin{lemma}   \label{lemma:equivariance_symbol_wise}
  A pair $(\sigma, \tau) \in \bij(\mathcal{X}) \times \bij(\mathcal{Y})$ is an equivariance of $p(Y|X)$ if and only if for all $(x,y) \in \mathcal{X} \times \mathcal{Y}$,
  \begin{align*}
    p(y|x) = p(\tau \cdot y | \sigma \cdot x).
  \end{align*}
\end{lemma}
\begin{proof}
  We have, writing $P_{Y|X}$ the column transition matrix corresponding to the channel $p(Y|X)$ and $G_{p(Y|X)}$ the equivariance group of $p(Y|X)$,
  \begin{align*}
    (\sigma, \tau) \in G_{p(Y|X)} \quad &\Leftrightarrow  \quad P_{Y|X} P_{\sigma} = P_{\tau} P_{Y|X} \\
    &\Leftrightarrow  \quad P_{Y|X} = P_{\tau} P_{Y|X} P_{\sigma^{-1}} \\
    &\Leftrightarrow  \quad P_{Y|X} = P_{\tau \cdot Y|\sigma \cdot X},
  \end{align*}
 where the last equivalence comes from the fact that the \emph{left} multiplication of $P_{Y|X}$ by the permutation matrix $P_\tau$ induces the permutation $\tau$ of the rows of $P_{Y|X}$; whereas the \emph{right} multiplication of $P_{Y|X}$ by the permutation matrix $P_{\sigma^{-1}}$ induces the permutation $(\sigma^{-1})^{-1} = \sigma$ of the columns of $P_{Y|X}$.
\end{proof}

\hlc{Now, as allowed by Theorem \ref{th:char_equivariance_with_iib_Ixy}'s assumption, we choose $p(X)$ such that $p(Y)$ is uniform. \hlcc{This implies, crucially, that $p(Y) = p(\tau \cdot Y)$, so that}
\begin{align*}
  p(y|x) = p(\tau \cdot y | \sigma \cdot x) \quad \Leftrightarrow \quad \frac{p(x,y)}{p(x)p(y)} = \frac{p(\sigma \cdot x, \tau \cdot y)}{p(\sigma \cdot x)p(\tau \cdot y)},
\end{align*}
i.e., recalling the definition of $\sim$ (see equation \eqref{eq:def_equivalence_relation}),
\begin{align*}
  p(y|x) = p(\tau \cdot y | \sigma \cdot x) \quad \Leftrightarrow \quad (x,y) \sim (\sigma \cdot x, \tau \cdot y).
\end{align*}
Taking Lemma \ref{lemma:equivariance_symbol_wise} into account, this yields:
\begin{proposition} \label{prop:equivariance_iff_equivalence}
  For a choice of $p(X)$ such that its image $p(Y)$ through the channel $p(Y|X)$ is uniform,
  \begin{align} \label{eq:equivariance_iff_equivalence}
    (\sigma, \tau) \in G_{p(Y|X)} \quad \Leftrightarrow \quad \forall (x,y) \in \mathcal{X} \times \mathcal{Y}, \ (x, y) \sim (\sigma \cdot x, \tau \cdot y).
  \end{align}
\end{proposition}
}

\subsection{Conclusion of the proof}
\label{apd:end_proof}

Here, we assume anew that $p(X,Y)$ is fully supported, i.e., that $\mathcal{S} = \mathcal{X} \times \mathcal{Y}$. 

Recalling that $(\sigma \otimes \tau) (x,y):= (\sigma \cdot x,\tau \cdot y)$ and that by definition of the deterministic clustering $\pi$ (see the beginning of Appendix \ref{apd:proof_charac_equivariance_with_iib_solutions}), we have $(x,y) \sim (x',y')$ if and only if $\pi(x,y) = \pi(x',y')$, we get that the right-hand-side in \eqref{eq:equivariance_iff_equivalence} is equivalent to
\begin{align}   \label{eq:local_proof_cor_1}
  \pi \circ (\sigma \otimes \tau) = \pi.
\end{align}
Now, from Corollary \ref{cor:solutions_to_IIB_fullsupport} and the fact that $p(X,Y)$ is fully supported, the solutions to the IIB for $\lambda = I(X;Y)$ are the channels of the form $\gamma \circ \pi$, for any congruent channel $\gamma \in C_{\text{cong}}(\{1, \dots, n\}, \mathcal{T})$. Thus, if we prove that, for any congruent channel $\gamma$, equation \eqref{eq:local_proof_cor_1} is equivalent to 
\begin{align}   \label{eq:local_proof_cor_2}
  \gamma \circ \pi \circ (\sigma \otimes \tau) = \gamma \circ \pi,
\end{align}
this would prove that for any solution $\kappa$ to the IIB for $\lambda = I(X;Y)$, we have $(\sigma, \tau) \in G_{p(Y|X)}$ if and only if \hl{$\kappa \circ (\sigma \otimes \tau) = \kappa$}: this is exactly the statement of Theorem \ref{th:char_equivariance_with_iib_Ixy}. Therefore, we only need to prove the following lemma:

\begin{lemma}
    Let $\mathcal{A}$, $\mathcal{B}$ and $\mathcal{C}$ be \hll{finite sets}. Consider two functions $f, g : \mathcal{A} \rightarrow \mathcal{B}$, and a congruent channel $\gamma \in C_{\text{cong}}(\mathcal{B}, \mathcal{C})$.
    Then $f = g$ if and only if $\gamma \circ f = \gamma \circ g$. 
\end{lemma}

\begin{proof}
  Clearly, $ f = g$ implies $\gamma \circ f = \gamma \circ g$. Conversely, assume that $\gamma \circ f = \gamma \circ g$. As $\gamma$ is congruent, the supports of the $\gamma(C|b)$, where $b \in \mathcal{B}$, are disjoint sets $\mathcal{C}_b \subseteq \mathcal{C}$. Let us consider the deterministic clustering $h \in C(\bigsqcup_{b \in \mathcal{B}} \mathcal{C}_b, \mathcal{B})$ defined by $h(b|c) := \delta_{c \in \mathcal{C}_b}$. Then $h \circ \gamma$ is the identity \hlcc{on} $\mathcal{B}$. But $\gamma \circ f = \gamma \circ g$ implies that
  \begin{align} 
    h \circ \gamma \circ f = h \circ \gamma \circ g,
  \end{align}
  which thus means exactly $f = g$.
\end{proof}

\hl{
\begin{remark}      \label{rmk:proof_for_general_pYgX}
    The only part of the proof where we used the full support assumption on $p(X,Y)$ was Appendix \ref{apd:end_proof}, which is thus the only part which would need\hlll{, in future work,} to be adapted to non-necessarily full-support distributions $p(X,Y)$.
\end{remark}
}

\section{\hll{Towards generalisations to non-finite variables}}
\label{apd:generalisation_to_nonfinite_case}

This work is set in the finite case, but it provides a basis for generalisations to more general settings. Indeed, the notions and tools used in this paper have straightforward generalisations to, for instance, the measure-theoretic setting --- which include finite, countable and continuous spaces. In particular, one can directly generalise\hlcc{, to Borel spaces \citep{rudinRealComplexAnalysis1987},} probabilities and conditional probabilities \citep{billingsleyProbabilityMeasure1995}, as well as the Kullback-Leibler divergence and mutual information \citep{grayEntropyInformationTheory2014}. Thus it seems that the IIB problem \eqref{eq:def_intertwining_ib} can be defined for Borel spaces. Moreover, the tools used in Appendix \ref{apd:explicit_solution_iib_Ixy} to describe explicitly the case $\lambda = I(X;Y)$ seem to adapt well to Borel spaces: namely, the log-sum inequality and its equality case; partitions induced by an equivalence relation; and the switching of the integration order for probability measures \citep{billingsleyProbabilityMeasure1995}.

Eventually, one can consider the action of measurable groups on Borel spaces \citep{kallenbergRandomMeasuresTheory2017}, along with the corresponding partition defined by the group action's orbits. One could thus consider measurable equivariances of conditional probabilities between Borel spaces. These concepts would allow the statement of Theorem \ref{th:char_equivariance_with_iib_Ixy} to be given a meaning in this general setting. \hlll{We leave to future work to fully adapt} the proof of Theorem \ref{th:char_equivariance_with_iib_Ixy} to \hlll{such a generalised statement.}

\end{document}